\documentclass[runningheads]{llncs}

\usepackage[T1]{fontenc}

\usepackage{amsmath}
\usepackage{amssymb}
\usepackage{mathpartir}
\usepackage{mathtools}
\usepackage{dsfont}
\usepackage{xcolor}
\usepackage{xspace}
\usepackage{hyperref}
\usepackage{lineno}
\usepackage{graphicx}
\usepackage{stmaryrd}
\usepackage{anyfontsize}
\usepackage{algorithm}
\usepackage{algpseudocode}
\usepackage{shortcuts}
\usepackage{cleveref}
\usepackage{booktabs}
\usepackage{listings}
\crefname{theorem}{Thm.}{Thms.}
\crefname{lemma}{Lem.}{Lemmas}
\crefname{corollary}{Cor.}{Cors.}
\crefname{figure}{Fig.}{Figs.}
\crefname{definition}{Defn.}{Defns.}
\crefname{table}{Tab.}{Tabs.}
\crefformat{section}{\S#2#1#3}
\crefmultiformat{section}{\S#2#1#3}{ and~\S#2#1#3}{, \S#2#1#3}{ and~\S#2#1#3}
\crefname{example}{Ex.}{Exs.}
\crefname{item}{item}{items}
\crefname{footnote}{footnote}{footnotes}
\crefname{observation}{Obs.}{Obs.}
\crefname{remark}{Remark}{Remarks}
\crefname{proposition}{Prop.}{Props.}
\crefname{equation}{Eqn.}{Eqns.}
\crefname{counterexample}{Counterexample}{Counterexamples}
\crefname{property}{Property}{Properties}
\crefname{algorithm}{Algorithm}{Algorithms}
\usepackage{subcaption}
\usepackage{setspace}

\newcommand{\Blue}[1]{\textcolor{blue}{#1}}
\newcommand{\Red}[1]{\textcolor{red}{#1}}
\newcommand{\Gray}[1]{\textcolor{gray}{#1}}

\renewcommand{\dagger}{\text{\textdagger}}

\newcommand{\QBUAd}{QBUA\textsuperscript{$\Diamond$}\xspace}
\newcommand{\Bd}{B\textsuperscript{$\Diamond$}}

\newcommand{\pjudge}[3]{\vdash\left\{#1\right\}#2\left\{#3\right\}}

\newcommand{\judge}[3]{\vdash\left[#1\right]#2\left[#3\right]}

\newcommand{\bjudge}[3]{\vdash_{\mathsf{B}}\left[#1\right]#2\left[#3\right]}
\newcommand{\fjudge}[3]{\vdash_{\mathsf{F}}\left[#1\right]#2\left[#3\right]}
\newcommand{\ajudge}[3]{\vdash_{\dagger}\left[#1\right]#2\left[#3\right]}
\newcommand{\djudge}[3]{\vdash^{\Diamond}_{\mathsf{B}}\left[#1\right]#2\left[#3\right]}
\newcommand{\bJudge}[3]{\vDash_{\mathsf{B}}\left[#1\right]#2\left[#3\right]}
\newcommand{\fJudge}[3]{\vDash_{\mathsf{F}}\left[#1\right]#2\left[#3\right]}
\newcommand{\dJudge}[3]{\vDash^{\Diamond}_{\mathsf{B}}\left[#1\right]#2\left[#3\right]}

\newcommand{\aajudge}[5]{\vdash_{\dagger}\left[\left.#1\right|#2\right]#3\left[\left.#4\right|#5\right]}

\newcommand{\AAjudge}[5]{\vdash_{\text{\textdaggerdbl}}\left[\left.#1\right|#2\right]#3\left[\left.#4\right|#5\right]}

\newcommand{\bigstep}[6]{\left\langle#1,#2,#3\right\rangle\Downarrow^{#4}\left\langle#5,#6\right\rangle}
\newcommand{\bigstepp}[5]{\left\langle#1,#2,#3\right\rangle\Downarrow\left\langle#4,#5\right\rangle}
\newcommand{\bigstepl}[5]{\left\langle#1,#2,#3\right\rangle\Downarrow^{\le 0}\left\langle#4,#5\right\rangle}

\newcommand{\eval}[2]{\left\llbracket#1\right\rrbracket#2}

\newcommand{\post}[2]{\mathrm{sp}\left\llbracket#1\right\rrbracket\left(#2\right)}
\newcommand{\pre}[2]{\mathrm{wp}\left\llbracket#1\right\rrbracket\left(#2\right)}
\newcommand{\prel}[2]{\mathrm{wp}^{\le 0}\left\llbracket#1\right\rrbracket\left(#2\right)}

\newcommand{\Skip}{\textsf{skip}}
\newcommand{\Assign}[2]{#1\coloneqq#2}
\newcommand{\Assume}[1]{\textsf{assume}~#1}
\newcommand{\Tick}[1]{\textsf{tick}~#1}
\newcommand{\Seq}[2]{#1;#2}
\newcommand{\Choice}[2]{#1+#2}
\renewcommand{\Loop}[1]{#1^*}
\newcommand{\Ite}[3]{\textsf{if}~#1~\textsf{then}~#2~\textsf{else}~#3}

\renewcommand{\While}[2]{\textsf{while}~#1~\textsf{do}~#2}

\newcommand{\Local}[2]{\textsf{local}~#1~\textsf{in}~#2}

\newcommand{\Fv}{\mathrm{fv}}
\newcommand{\Mod}{\mathrm{mod}}

\newcommand{\Sup}{\mathop{\scalebox{-1}[1]{\textsf{S}}}}
\newcommand{\Inf}{\mathop{\scalebox{-1}[1]{\textsf{J}}}}

\newcommand{\true}{\mathsf{true}}
\newcommand{\false}{\mathsf{false}}

\algdef{SE}[SUBALG]{Indent}{EndIndent}{}{\algorithmicend\ }%
\algtext*{Indent}
\algtext*{EndIndent}

\begin{document}


\title{A Program Logic for Under-approximating Worst-case Resource Usage}

\author{Ziyue Jin \and Di Wang\thanks{Corresponding Author}}
\institute{Key Lab of HCST (PKU), MOE; SCS, Peking University, China}

%

\maketitle

\begin{abstract}
%
Understanding and predicting the worst-case resource usage is crucial for software quality; however, existing methods either over-approximate with potentially loose bounds or under-approximate without asymptotic guarantees.
This paper presents a program logic to under-approximate worst-case resource usage, adapting incorrectness logic (IL) to reason quantitatively about resource consumption.
We propose \underline{q}uantitative \underline{f}orward and \underline{b}ackward \underline{u}nder-\underline{a}pproximate (QFUA and QBUA) triples, which generalize IL to identify execution paths leading to high resource usage.
We also introduce a variant of QBUA that supports reasoning about high-water marks.
Our logic is proven sound and complete with respect to a simple IMP-like language, and all meta-theoretical results are mechanized and verified in Rocq.
We implement a prototype checker for all three variants of our logic and demonstrate its utility through a few examples and four case studies.
%

\end{abstract}

\section{Introduction}
\label{sec:intro}


%
%
%

Understanding and predicting the worst-case resource usage of programs is a fundamental challenge in computer science.
Knowing the bounds of resource usage, whether for memory consumption, CPU cycles, or network bandwidth, is crucial for ensuring system reliability, performance, and security.
However, precisely determining these bounds is notoriously difficult.
Static methods for analyzing resource usage tend to \emph{over-}approximate worst-case scenarios, resulting in non-tight resource bounds and, consequently, false positives.
On the other hand, dynamic methods, such as testing, can typically produce specific inputs of a certain size to induce high resource usage---thus \emph{under-}approximating worst-case scenarios---but do not offer an asymptotic characterization of the worst-case behavior.
This paper aims to develop a program logic to under-approximate worst-case resource usage and offer a compositional method for identifying scenarios with high resource usage.

Previous approaches to resource analysis have mainly concentrated on over-approximation methods, including abstract interpretation and constraint-based techniques.
These techniques aim to compute sound upper bounds on resource usage by taking into account all possible program behaviors.
However, they may produce non-tight bounds and typically do not indicate which inputs would result in the worst-case resource usage.
Under-approximation techniques, such as fuzzing and dynamic analysis, have been used to identify specific resource-intensive execution paths.
However, these techniques lack compositionality and generality for comprehensive resource analysis.
Incorrectness logic (IL), which has been successfully applied to bug detection, offers a promising alternative by providing a formal foundation for under-approximate reasoning~\cite{POPL:OHearn20}.
Recent work has adapted IL to prove non-termination~\cite{OOPSLA:RVO24}, i.e., a \emph{qualitative} argument about resource usage.
However, adapting IL to prove \emph{quantitative} resource bounds has not been explored.

In this paper, we adapt IL to under-approximate worst-case resource usage.
The key idea is to use under-approximate triples to reason about the existence of execution paths that lead to high resource usage.
Specifically, we introduce a new form of under-approximate triple that captures the relationship between program states and resource consumption.
Our approach leverages IL's compositional nature to analyze a program by breaking it down into smaller, manageable components.

One of the main challenges is to formalize resource usage within the IL framework.
The original IL triple $\fjudge{p}{C}{ok:q}$---also called a \emph{forward under-approximate} (FUA) triple---denotes that \emph{every} post-state that satisfies $q$ is \emph{reachable} by executing $C$ from \emph{some} pre-state that satisfies $p$.
Raad et al.~\cite{OOPSLA:RVO24} suggested using \emph{backward under-approximate} (BUA) triples when reasoning about non-termination; that is, $\bjudge{p}{C}{ok:q}$ denotes that from \emph{every} pre-state that satisfies $p$ is possible to \emph{reach} \emph{some} post-state that satisfies $q$ by executing $C$.
The pre-conditions $p$ and post-conditions $q$ are \emph{qualitative}, i.e., they are Boolean-valued assertions on program states; however, they do not provide a natural way to reason about resource consumption.
In our work, we need to extend IL to capture the relationship between program states and resource usage.
We adapt the idea of \emph{quantitative Hoare logic}~\cite{PLDI:CHR14,PLDI:CHS15}: instead of the Boolean $true$, a quantitative assertion returns a natural number that indicates the amount of resource that is required to safely execute the program; at the same time, the Boolean $false$ is encoded by $\infty$.
Intuitively, quantitative Hoare logic \emph{refines} classic Hoare logic to reason about both the functionality and the resource usage of a program.
In this paper, we devise \emph{quantitative} generalizations of IL:
let $P$ and $ Q$ be assertions of type $State \to \bbZ \cup \{-\infty, +\infty\}$;
the FUA triple $\fjudge{P}{C}{Q}$ denotes that every post-state with \emph{at least} $Q$ units of resource is reachable by executing $C$ from some pre-state with \emph{at least} $P$ units of resource;
and the BUA triple $\bjudge{P}{C}{Q}$ denotes that every pre-state with \emph{at most} $P$ units of resource is possible to reach some post-state with \emph{at most} $Q$ units of resource by executing $C$.
In addition, we adapt Carbonneaux et al.~\cite{PLDI:CHS15}'s approach to separate qualitative assertions and quantitative resource bounds in the FUA/BUA triples, as well as Zhang and Kaminski's approach~\cite{OOPSLA:ZK22} to calculate quantitative strongest post-conditions, to simplify the reasoning in our prototype implementation and case studies.

Another challenge is to support reasoning about \emph{high-water marks} in our quantitative IL.
High-water marks provide a more precise characterization of a program's resource usage when the resource is non-monotone.
For example, when reasoning about stack-space bounds, we want to know the highest stack watermark \emph{during} a program's execution.
However, an FUA/BUA triple $\ajudge{P}{C}{Q}$ only constrains the amount of resource at the pre- and post-states;
that is, although the difference between $P$ and $Q$ serves as an under-approximation of the worst-case resource consumption of $C$, we cannot say that $P$ provides an under-approximation of the high-water mark of executing $C$.
It is worth noting that for quantitative Hoare logic $\vdash \left\{P\right\} C \left\{Q\right\}$, $P$ is indeed an over-approximation of the high-water mark.
In this paper, we devise a variant of the BUA logic to under-approximate high-water marks:
the triple $\djudge{P}{C}{Q}$ means that every pre-state with at most $P$ units of resource is possible to reach some post-state with at most $Q$ units of resource by executing $C$, \emph{and the particular execution must make the resource counter non-positive at some point}.
In this way, $P$ can serve as an under-approximation of the high-water mark of $C$ because, if we start to execute $C$ with $P$ units of resource, there exists an execution that consumes all the $P$ units at some point.
Note that different from quantitative Hoare logic, our quantitative IL allows the quantitative assertions $P,Q$ to take negative values, i.e., they have type $State \to \bbZ \cup \{\pm\infty\}$.

In this paper, we focus on the theoretical properties of our quantitative IL for under-approximating worst-case resource usage.
We prove that both the FUA and BUA variants are sound and complete with respect to a resource-aware operational semantics for integer IMP programs.
To demonstrate the usefulness of our quantitative IL, we implement a prototype checker for our logic that supports arrays via the array theory and present several case studies, including some array-based sorting algorithms.
In the future, we plan to automate our logic by integrating with existing tools that target realistic languages, e.g., Pulse~\cite{OOPSLA:LRV22}.

%
In this paper, we make the following three main contributions.
\begin{itemize}
  \item We devise a program logic to under-approximate worst-case resource usage of programs. Our logic is a generalization of the incorrectness logic, and we formulate three variants: forward, backward, and backward high-water mark.
  \item We prove that our program logic (both the forward and the backward variants) is sound and complete with respect to a resource-aware operational semantics, with all proofs mechanized in Rocq.
  \item We implement a prototype checker for all three variants of our logic and present four case studies to demonstrate the usefulness of it.
\end{itemize}

\section{Overview}
\label{sec:overview}



We first examine quantitative Hoare logic to introduce the idea of resource-aware program logics (\cref{sec:overview:qhl}).
Next, we review incorrectness logic and non-termination logic to discuss the concept of under-approximation (\cref{sec:overview:il}).
Finally, we outline our approach to developing quantitative under-approximate logic to reason about worst-case resource usage (\cref{sec:overview:qual}).

\subsection{Prior Work: Quantitative Hoare Logic}
\label{sec:overview:qhl}

Our work is not the first to consider developing program logics to derive resource-usage bounds.
We take inspiration from the \emph{quantitative Hoare logic} (QHL), which extends classic Hoare logic by incorporating resource consumption into program reasoning~\cite{PLDI:CHR14,PLDI:CHS15}.

The key idea behind QHL is to augment Hoare triples with resource annotations.
A \emph{quantitative Hoare triple}, written as $\pjudge{P}{C}{Q}$ where $P$ and $Q$ are quantitative assertions of type $State \to \bbN \cup \{ \infty\}$, states that if $C$ is executed in a pre-state with \emph{at least} $P$ units of resource, then after its execution, \emph{at least} $Q$ units of resource will remain.
Intuitively, these quantitative assertions can be seen as \emph{potential functions} that map program states to non-negative potentials: the pre-potential $P$ is sufficient to pay for the resource usage of $C$ as well as the post-potential $Q$.
In other words, QHL essentially captures the principle of the \emph{potential method} for amortized complexity analysis~\cite{JADM:Tarjan85}.

Below are three representative rules of QHL.
To explicitly annotate resource consumption, people usually introduce a primitive command $\Tick{e}$ that computes $e$ to an integer $n$ and then consumes $n$ units of resource.
(If $n$ is negative, this tick command releases $-n$ units of resource.)
The rule \textsc{(QHL:Tick)} states that if the pre-potential is at least $P + e$, then it is safe to execute $\Tick{e}$ and end with post-potential $P$.
The rule \textsc{(QHL:Seq)} illustrates the \emph{compositional} nature of QHL: to reason about the resource usage of the sequencing command $\Seq{C_1}{C_2}$, one can derive quantitative triples \emph{individually} for $C_1$ and $C_2$.
The rule \textsc{(QHL:While)} generalizes the invariant-based reasoning of classic Hoare logic, where the predicate $\mathsf{istrue}(B)$ encodes a potential function that returns $0$ if $B$ is true and $\infty$ otherwise.
\begin{linenomath}
\begin{mathpar}\small
  \inferrule[(QHL:Tick)]
  { P + e \ge 0 }
  { \pjudge{P + e}{\Tick{e}}{P} }
  \quad
  \inferrule[(QHL:Seq)]
  { \pjudge{P}{C_1}{R} ~ \pjudge{R}{C_2}{Q} }
  { \pjudge{P}{\Seq{C_1}{C_2}}{Q} }
  \quad
  \inferrule[(QHL:While)]
  { \pjudge{I + \mathsf{istrue}(B)}{C}{I} }
  { \pjudge{I}{\While{B}{C}}{I + \mathsf{istrue}(\neg B)} }
\end{mathpar}
\end{linenomath}

Consider the program $\While{x<n}{(\Seq{\Assign{x}{x+1}}{\Tick{1}})}$.
Intuitively, the total resource consumption is given by $I \coloneqq \max(0, n - x)$.
Using $I$ as the loop invariant, we can prove the judgement $\pjudge{I}{\While{x<n}{(\Assign{x}{x+1};\Tick{1})}}{0}$.
By \textsc{(QHL:While)}, this amounts to proving $\pjudge{I+\mathsf{istrue}(x<n)}{\Assign{x}{x+1}; \Tick{1}}{I}$.
If $x \ge n$, then the pre-potential is $\infty$ and thus the triple is valid;
otherwise, we can use \textsc{(QHL:Seq)} with $\max(0,n-x) + 1$ as the intermediate assertion and conclude by $\max(0,n-(x+1))+1 = \max(0,n-x)$ when $x<n$.

It is worth noting that QHL also supports reasoning about \emph{high-water marks}.
In fact, the quantitative triple $\pjudge{P}{C}{Q}$ indicates that $P$ is an upper bound (i.e., an over-approximation) of the high-water mark of executing $C$.
Consider the non-terminating program $\While{\mathsf{true}}{(\Tick{1}; \Tick{{-1}})}$.
We can prove the judgement $\pjudge{1}{\While{\mathsf{true}}{(\Tick{1};\Tick{{-1}})}}{1}$, i.e., the high-water mark of the program is upper-bounded by $1$.
Again, by \textsc{(QHL:While)}, this amounts to proving $\pjudge{1}{\Tick{1}; \Tick{{-1}}}{1}$.
This can be justified by using \textsc{(QHL:Seq)} with $0$ as the intermediate assertion.
Note that the triple also indicates that after the program terminates, the resource consumption is upper-bounded by $1-1=0$.
However, because the program is non-terminating, the triple is indeed an over-approximation in terms of resource consumption.

It is important to note that while QHL focuses on proving sound upper bounds (i.e., over-approximations), the goal of this paper is to establish sound \emph{under-approximations} of resource usage, as we will further elaborate in \cref{sec:overview:qual}.

\subsection{Prior Work: Incorrectness Logic and Non-termination Logic}
\label{sec:overview:il}

Incorrectness logic (IL), introduced by O'Hearn~\cite{POPL:OHearn20}, provides a formal foundation for under-approximate reasoning about program behaviors, particularly for bug detection.
Unlike classic Hoare logic, which focuses on proving program correctness by over-approximating program behaviors, IL under-approximates program behaviors to identify specific execution paths that lead to bugs.
This makes IL well-suited for detecting errors such as memory safety violations~\cite{CAV:RBD20}, concurrency bugs~\cite{POPL:RBD22}, and non-termination~\cite{OOPSLA:RVO24}, where the goal is to find concrete evidence of incorrect behavior rather than proving the absence of errors.

The key idea behind IL is the use of \emph{forward, under-approximate} (FUA) triples, written as $\fjudge{p}{C}{\epsilon : q}$, which state that starting from a set of pre-states $p$, executing program $C$ can lead to a set of post-states under-approximated by $q$ under the exit condition $\epsilon$ (e.g., normal termination $ok$ or error $er$).
The under-approximate nature of FUA triples ensures that any bug detected is a true positive, as it corresponds to a concrete execution path that exhibits the error.
Below are two rules for under-approximating the behavior of while loops in IL.
Different from invariant-based reasoning in classic Hoare logic, IL essentially \emph{unrolls} a loop to examine a subset of execution paths through the loop.
The rule \textsc{(IL:WhileFalse)} states that if the loop condition $B$ does not hold before entering the loop, the pre-states and post-states coincide because the loop is not entered.
The rule \textsc{(IL:WhileSubvar)} provides \emph{subvariant}-based reasoning, which generalizes bounded unrolling by asserting that there exists some $k$ satisfying that the loop body $C$ transforms $p(n) \wedge B$ to $p(n+1) \wedge B$ for $n<k$, and that $C$ transforms $p(k) \wedge B$ to $q \wedge \neg B$, i.e., the loop could end after $k$ iterations.
\begin{linenomath}
\begin{mathpar}\small
  \inferrule[(IL:WhileFalse)]
    { }
    { \fjudge{p \wedge \neg B}{\While{B}{C}}{ok: p \wedge \neg B} }
    \and
    \inferrule[(IL:WhileSubvar)]
    { \forall n < k. \fjudge{p(n) \wedge B}{C}{ok: p(n+1) \wedge B} \\\\
      \fjudge{p(k) \wedge B}{C}{\epsilon: q \wedge \neg B}
    }
    { \fjudge{p(0) \wedge B}{\While{B}{C}}{\epsilon: q \wedge \neg B} }
\end{mathpar}
\end{linenomath}

Among the prior extensions of IL, the \emph{under-approximate, non-termination logic} (\textsc{UNTer}) is the most related work because it reasons---qualitatively---about the usage of a particular kind of resource, i.e., time~\cite{OOPSLA:RVO24}.
\textsc{UNTer} introduces divergent triples of the form $\judge{p}{C}{\infty}$, which state that starting from any state in $p$, the program $C$ has at least one divergent (i.e., non-terminating) execution.
A key insight in \textsc{UNTer} is the use of \emph{backward, under-approximate} (BUA) triples, written as $\bjudge{p}{C}{\epsilon: q}$, which state that starting from a set of pre-states $p$, executing program $C$ can reach some post-states in $q$ under the exit condition $\epsilon$.
Consider the program $\Assign{x}{x-1}$.
The FUA judgement $\fjudge{x>0}{\Assign{x}{x-1}}{ok: x>0}$ is valid, because the precise set of post-states is $x \ge 0$ and $x > 0$ is its subset, i.e., an under-approximation.
On the other hand, the BUA judgement $\bjudge{x>0}{\Assign{x}{x-1}}{ok: x > 0}$ is \emph{invalid}, because executing from the pre-state $x=1$ results in $x=0$, which is not included in the set $x>0$.
\textsc{UNTer} uses BUA triples to reason about the non-termination of loops.
For example, the rule \textsc{(UNTer:While)} shown below states that if executing the loop body $C$ keeps the states $p$ and the loop condition $B$ unchanged, we can establish that the loop can diverge.
If the rule used FUA triples instead, it would prove the triple $\judge{x > 0}{\While{x>0}{\Assign{x}{x-1}}}{\infty}$, which is unsound.
Another rule \textsc{(UNTer:WhileSubvar)} is similar to \textsc{(IL:WhileSubvar)}, providing subvariant-based reasoning about divergent loops.
\begin{linenomath}
\begin{mathpar}\small
  \inferrule[(UNTer:While)]
  { \bjudge{p \wedge B}{C}{ok: p \wedge B} }
  { \judge{p \wedge B}{\While{B}{C}}{\infty} }
  \and
  \inferrule[(UNTer:WhileSubvar)]
  { \forall n \in \bbN. \bjudge{p(n) \wedge B}{C}{p(n+1) \wedge B} }
  { \judge{p(0) \wedge B}{\While{B}{C}}{\infty} }
\end{mathpar}
\end{linenomath}


\subsection{This Work: Quantitative Under-approximate Logic}
\label{sec:overview:qual}

Our goal is to reason about the \emph{worst-case resource usage} of programs from an \emph{under-approximate} perspective.
Indeed, this can in principle be achieved within Incorrectness Logic (IL) by adding a ghost variable that tracks resource consumption.
For example, one may introduce a counter $ticks$ that is decreased by each $\Tick{e}$ command and reason about triples of the form $\fjudge{ticks\ge P}{C}{ticks\ge Q}$, expressing that there exists an execution of $C$ consuming at least $P-Q$ units of resource.
However, such an encoding is less compositional and modular.
For instance, from $\fjudge{ticks\ge 10}{C_1}{ticks\ge 0}$ and $\fjudge{ticks\ge 5}{C_2}{ticks\ge 0}$ one cannot directly derive a triple for $\Seq{C_1}{C_2}$ without re-establishing intermediate assertions such as $\fjudge{ticks\ge 15}{C_1}{ticks\ge 5}$.

To obtain a more compositional and elegant approach, we draw inspiration from quantitative Hoare Logic and lift incorrectness reasoning itself to the quantitative setting, devising \emph{quantitative under-approximate logic}.
Recall that in qualitative under-approximate logic, forward and backward triples capture two complementary reasoning modes: the backward form identifies possible initial states that can lead to certain final states, while the forward form discovers possible final states that can arise from some initial state.
Generalizing these two directions to the quantitative setting, we obtain \emph{quantitative, forward/backward, under-approximate} (QFUA/QBUA) triples that reason not only about reachability but also about the amount of resource consumed along executions.
Intuitively, QFUA is outcome-oriented, focusing on verifying that certain outputs can arise from executions with high resource usage, whereas QBUA is input-oriented, identifying initial states that can lead to such costly behaviors.
Formally:
\begin{itemize}
  \item A QFUA judgement $\fjudge{P}{C}{Q}$ states that every post-potential \emph{no less than} $Q$ is reachable by executing $C$ from some pre-potential with \emph{no less than} $P$.
  Thus, if we fix the post-potential to $Q$, the judgement states that there exists an execution path with some pre-potential $P' \ge P$, so that the difference between $P$ and $Q$ serves as an under-approximation of the worst-case cost, which is lower-bounded by the difference between $P'$ and $Q$.
  \item A QBUA judgement $\bjudge{P}{C}{Q}$ states that every pre-potential \emph{no greater than} $P$ is possible to reach some post-potential \emph{no greater than} $Q$ by executing $C$.
  If we fix the pre-potential to $P$, the judgement states that there exists an execution path with some post-potential $Q' \le Q$, so that the difference between $P$ and $Q$ also serves as an under-approximation of the worst-case scenario, which is lower-bounded by the difference between $P$ and $Q'$.
\end{itemize}

Different from \textsc{UNTer}~\cite{OOPSLA:RVO24}, where only backward triples are meaningful for non-termination reasoning, we will prove in \cref{sec:technical} that both QFUA and QBUA logics are sound and complete for under-approximating worst-case resource usage.
In our setting, we allow quantitative assertions to have type $State \to \bbZ \cup \{\pm\infty\}$.

To illustrate how the forward and backward reasoning modes complement each other, consider the program $\Ite{x = 42}{(\Seq{\Tick{2}}{\Assign{x}{0}})}{\Tick{1}}$.
Suppose we are interested in whether a final state with $x = 0$ can result from an execution that consumes at least 2 units of resource.
This can be expressed using a QFUA triple
\[
\fjudge{2}{\Ite{x = 42}{(\Seq{\Tick{2}}{\Assign{x}{0}})}{\Tick{1}}}{\max(0,[x \ne 0])},
\]
where the predicate $[B]$ encodes a potential function that returns $+\infty$ if $B$ is true and $-\infty$ otherwise.
It shows that some execution reaching a post-state with $x = 0$ requires at least 2 units of resource, indicating the existence of a high-cost final configuration.
On the other hand, if we wish to determine whether there exists an input that causes this high usage, we use a QBUA triple:
\[
\bjudge{\min(2, [x = 42])}{\Ite{x = 42}{(\Seq{\Tick{2}}{\Assign{x}{0}})}{\Tick{1}}}{0}, 
\]
which states that for any initial state satisfying $x = 42$ and having at most 2 units of resource, it is possible to execute the program and reach a post-state with at most 0 units remaining.
This example illustrates the complementary nature of the two logics: QFUA helps confirm that certain costly outcomes are possible, while QBUA helps identify the inputs that lead to them.

Below are some rules to derive the above QFUA and QBUA triples.
The \textsc{(F/B:IfTrue)} rules reflect the difference between QFUA and QBUA:
the $\max$ operator reflects QFUA's ``no less than'' nature, whereas
the $\min$ operator reflects QBUA's ``no greater than'' nature.
\begin{linenomath}
\begin{mathpar}\small
  \inferrule[(F:Tick)]
  { }
  { \fjudge{P}{\Tick{e}}{P-e} }
  \and
  \inferrule[(F:IfTrue)]
  { \fjudge{\max(P, [\neg B])}{C_1}{Q} }
  { \fjudge{\max(P, [\neg B])}{\Ite{B}{C_1}{C_2}}{Q} }  
  \\
  \inferrule[(B:Tick)]
  { }
  { \bjudge{P}{\Tick{e}}{P-e} }
  \and
  \inferrule[(B:IfTrue)]
  { \bjudge{\min(P, [B])}{C_1}{Q} }
  { \bjudge{\min(P, [B])}{\Ite{B}{C_1}{C_2}}{Q} }  
\end{mathpar}
\end{linenomath}

In contrast to the ghost-variable encoding, our logic admits direct composition through the \textsc{(Relax)} rule, which allows shifting both pre- and post-potentials uniformly without re-verifying the body of the command:
\begin{linenomath}
\begin{mathpar}\small
  \inferrule[(\dagger:Relax)]
  { \ajudge{P}{C}{Q} \\ F \text{ is invariant under } C }
  { \ajudge{P+F}{C}{Q+F} },
  \quad \text{where~}\dagger \in \{\mathsf{F}, \mathsf{B}\}
\end{mathpar}
\end{linenomath}
This rule restores the compositionality that is lost in the ghost-variable encoding.
For example, recall the non-compositional case discussed at the beginning of this section, where the triples $\fjudge{10}{C_1}{0}$ and $\fjudge{5}{C_2}{0}$ could not be directly combined.
By applying the \textsc{(F:Relax)} rule, we can derive $\fjudge{15}{C_1}{5}$ directly and then combine them via the sequential composition rule to obtain $\fjudge{15}{C_1;C_2}{0}$, without re-verifying $C_1$.
This demonstrates how our framework supports modular reasoning about resource usage, addressing the compositionality issue mentioned earlier.
We will see further examples of such compositional reasoning in~\cref{sec:case-studies}.

Constructing an equivalent rule for a ghost-variable encoding, however, poses a dilemma.
If a ghost-variable rule were designed to only handle pre- and post-conditions of the specific form $ticks \ge \cdots$, it would lack the generality to handle more complex state-resource relationships.
Conversely, if the rule were designed to accommodate arbitrary pre- and post-conditions, defining the additive operation ($P+F$) used in our \textsc{(Relax)} rule would become non-trivial.
Thus, our logics conceptually achieve modularity in a more systematical way.

Recall that a triple in quantitative Hoare logic also provides an over-approximation of the high-water mark for executing a program.
However, this is not the case for QFUA or QBUA triples.
For example, we can derive the QBUA judgement $\bjudge{\min(20,[x>y])}{\Ite{x>y}{\Tick{2}}{\Tick{1}}}{18}$, but $\min(20,[x>y])$ is obviously \emph{not} a lower bound of the worst-case high-water mark for executing the conditional command.
To support under-approximating high-water marks, we devise a third kind of quantitative under-approximate triples, which we write as $\djudge{P}{C}{Q}$ and call \QBUAd triples\footnote{We discuss why the forward variant QFUA\textsuperscript{$\Diamond$} is not developed in~\cref{sec:technical:hwm}.}.
Such a triple is a refinement of $\bjudge{P}{C}{Q}$; that is, every pre-potential no greater than $P$ is possible to reach some post-potential no greater than $Q$ by executing $C$, \emph{and during the particular execution, the potential must become non-positive at some point}.
We prove in \cref{sec:technical} that $\djudge{P}{C}{Q}$ indicates that $P$ serves as an under-approximation of the high-water mark for executing $C$.
In addition, the \QBUAd logic is also complete.
Below are the $\Diamond$ versions of the rules shown earlier.
\begin{linenomath}
\begin{mathpar}\small
  \inferrule[(\Bd:Tick)]
  { \min(P,P-e) \le 0 }
  { \djudge{P}{\Tick{e}}{P-e} }
  \and
  \inferrule[(\Bd:IfTrue)]
  { \djudge{\min(P, [B])}{C_1}{Q} }
  { \djudge{\min(P, [B])}{\Ite{B}{C_1}{C_2}}{Q} }  
\end{mathpar}
\end{linenomath}
In our formulation, the \QBUAd logic depends on the QBUA logic.
For example, there are two possibilities for reasoning about the sequencing command $C_1;C_2$: executing $C_1$ makes the potential non-positive, or executing $C_2$ makes the potential non-positive, as shown below.
\begin{linenomath}
\begin{mathpar}\small
  \inferrule[(\Bd:SeqL)]
  { \djudge{P}{C_1}{R} \\ \bjudge{R}{C_2}{Q} }
  { \djudge{P}{C_1;C_2}{Q} }
  \and
  \inferrule[(\Bd:SeqR)]
  { \bjudge{P}{C_1}{R} \\ \djudge{R}{C_2}{Q} }
  { \djudge{P}{C_1;C_2}{Q} }
\end{mathpar}
\end{linenomath}

\section{Technical Details}
\label{sec:technical}



In this section, we present the formal systems for the QFUA, QBUA, and \QBUAd logics for under-approximating worst-case resource consumption.
\cref{sec:technical:lang} formulates an IMP-style programming language with integer variables and tick commands.
\cref{sec:technical:qual} describes the QFUA and QBUA logics and proves their soundness and completeness.
\cref{sec:technical:hwm} extends QBUA to develop \QBUAd, which aims to under-approximate high-water marks.
All metatheoretic results presented in this section have been mechanically verified in Rocq; the full formalization is available in the artifact~\cite{software:JW26}.

\subsection{Syntax and Semantics}
\label{sec:technical:lang}

The language syntax defines resource-aware computation through the grammar:
\begin{align*}
C ::= \Skip \mid \Assign{x}{e} \mid \Assume{B} \mid \Tick{e} 
     \mid \Seq{C}{C} \mid \Choice{C}{C} \mid \Loop{C} \mid \Local{x}{C}
\end{align*}
The basic operations include: $\Skip$ denoting no-operation, $\Assign{x}{e}$ for variable assignments, $\Assume{B}$ constraining execution to paths satisfying boolean condition $B$, and $\Tick{e}$ modeling resource consumption by deducting the value of arithmetic expression $e$ from an implicit resource counter.
Control flow is structured through $\Seq{C_1}{C_2}$ for sequencing, $\Choice{C_1}{C_2}$ for non-deterministic choice, and $\Loop{C}$ for iterative execution. Variable scoping is managed via $\Local{x}{C}$.
All program variables are integer-valued.
Derived control structures are defined as: conditional execution $\Ite{B}{C_1}{C_2} \coloneqq \Choice{\Seq{(\Assume{B}}{C_1})}{(\Seq{\Assume{(\neg B)}}{C_2})}$ and looping constructs $\While{B}{C} \coloneqq \Seq{\Loop{\left(\Seq{\Assume{B}}{C}\right)}}{\Assume{(\neg B)}}$.

\begin{figure}[t!]
\begin{mathpar}\footnotesize
  \inferrule[(BS:Tick)]{}{\bigstep{\Tick{e}}{\sigma}{p}{\min\{p,p-\eval{e}{\sigma}\}}{\sigma}{p-\eval{e}{\sigma}}}
  \hva\and
  \inferrule[(BS:Seq)]{\bigstep{C_1}{\sigma}{p}{l_1}{\rho}{r} \\ \bigstep{C_2}{\rho}{r}{l_2}{\tau}{q}}{\bigstep{\Seq{C_1}{C_2}}{\sigma}{p}{\min\{l_1,l_2\}}{\tau}{q}}
\end{mathpar}
\caption{Selected Rules for the Big-step Semantics}
\label{fig:selectedsemantics}
\end{figure}

The big-step semantics formalizes program behavior with explicit resource tracking.
The judgment $\bigstep{C}{\sigma}{p}{l}{\tau}{q}$ states that executing $C$ from initial state $\sigma \in State \coloneqq Var \to \bbZ$ with resource $p \in \bbZ$ terminates in state $\tau \in State$ with residual resource $q \in \bbZ$, where $l \in \bbZ$ records the minimal resource level observed during execution.
The semantics rules are selectively presented in \cref{fig:selectedsemantics}, with the full set of rules provided in the Appendix (\cref{fig:fullsemantics}).
We also write $\bigstepp{C}{\sigma}{p}{\tau}{q}$ to denote $\exists l.\bigstep{C}{\sigma}{p}{l}{\tau}{q}$, and $\bigstepl{C}{\sigma}{p}{\tau}{q}$ to denote $\exists l\le 0.\bigstep{C}{\sigma}{p}{l}{\tau}{q}$.


\subsection{Quantitative Under-approximate Logic}
\label{sec:technical:qual}

We generalize both forward and backward under-approximate logic to quantitative reasoning by extending assertions from Boolean predicates $State \to \{\true, \false\}$ to \emph{resource functions} $State \to \mathbb{Z} \cup \{\pm\infty\}$. These functions map program states to integer-valued resource quantities (with $\pm\infty$ denoting Boolean truth and falsehood). The generalized logics are called \emph{quantitative forward under-approximate} (QFUA) logic and \emph{quantitative backward under-approximate} (QBUA) logic, whose semantics are defined in terms of resource functions.

\begin{definition}
  The semantics of QFUA and QBUA triples are defined as follows:
  \begin{itemize}
    \item QFUA: $\fJudge{P}{C}{Q}$ holds if and only if for all $\tau$ and $q$ such that $q\ge Q(\tau)$, there exists a $\sigma$ and a $p$ such that $p\ge P(\sigma)$ and $\bigstepp{C}{\sigma}{p}{\tau}{q}$. This implies that for all $\tau$ such that $Q(\tau)\in\mathbb{Z}$, there exists an execution of $C$ ending in $\tau$ that cost at least $\inf_\sigma\{P(\sigma)\}-Q(\tau)$ ticks.
    \item QBUA: $\bJudge{P}{C}{Q}$ holds if and only if for all $\sigma$ and $p$ such that $p\le P(\sigma)$, there exists a $\tau$ and a $q$ such that $q\le Q(\tau)$ and $\bigstepp{C}{\sigma}{p}{\tau}{q}$. This implies that for all $\sigma$ such that $P(\sigma)\in\mathbb{Z}$, there exists an execution of $C$ starting from $\sigma$ that cost at least $P(\sigma)-\sup_\tau\{Q(\tau)\}$ ticks.
  \end{itemize}
\end{definition}

\begin{figure}[t!]
\begin{mathpar}\footnotesize
  P\curlywedge Q\coloneqq\lambda\sigma.\min\{P(\sigma),Q(\sigma)\}
  \hva\and
  P\curlyvee Q\coloneqq\lambda\sigma.\max\{P(\sigma),Q(\sigma)\}
  \hva\and
  \Sup x.P\coloneqq\lambda\sigma.\sup_v\{P(\sigma[x\mapsto v])\}
  \hva\and
  \Inf x.P\coloneqq\lambda\sigma.\inf_v\{P(\sigma[x\mapsto v])\}
  \hva\and
  P\preceq Q\quad\text{iff}\quad\forall\sigma.P(\sigma)\leq Q(\sigma)
  \hva\and
  [B]\coloneqq\lambda\sigma.\begin{cases}+\infty&\text{if }B(\sigma)=\true\\-\infty&\text{otherwise}\end{cases}
\end{mathpar}
\caption{Resource Function Operators}
\label{fig:operators}
\end{figure}

To formulate quantitative verification rules, we define operators on resource functions as shown in \cref{fig:operators}. The $\curlyvee$ and $\curlywedge$ operators compute pointwise maximum and minimum of two functions respectively, while $\Sup x.P$ and $\Inf x.P$---adapted from Batz et al.~\cite{POPL:BKK21}'s work and Zhang and Kaminski~\cite{OOPSLA:ZK22}'s work---capture external values over variable assignments. The refinement operator $\preceq$ establishes pointwise ordering between functions, and $[B]$ converts Boolean predicates by mapping truth values to infinite bounds.

\begin{table}[t!]
\centering
\caption{Correspondence between Boolean and Quantitative Operators}
\label{fig:correspondence}

\vspace{0.5em}

\begin{tabular}{l@{\hspace{2em}}c@{\hspace{2em}}c@{\hspace{2em}}c}
\toprule
&Boolean&QFUA&QBUA\\
\midrule
Conjunction&$P\land Q$&$P\curlyvee Q$&$P\curlywedge Q$\\
Disjunction&$P\lor Q$&$P\curlywedge Q$&$P\curlyvee Q$\\
Existential&$\exists x.P$&$\Inf x.P$&$\Sup x.P$\\
Implication&$P\Rightarrow Q$&$P\preceq Q$&$Q\preceq P$\\
Boolean Predicate&$B$&$[\neg B]$&$[B]$\\
\bottomrule
\end{tabular}
\end{table}

These operators extend Boolean logic to integer-valued resource analysis. QFUA and QBUA triples use different interpretations of logical operators due to their opposite inequality directions in their semantics.
To illustrate this generalization, consider the case of conjunction: in Boolean logic, $P \land Q$ means both predicates must hold. For QFUA triples, this translates to $P \curlyvee Q$, as satisfying both $p \geq P(\sigma)$ and $p \geq Q(\sigma)$ is equivalent to $p \geq \max\{P(\sigma), Q(\sigma)\}$. Conversely, QBUA triples use $P \curlywedge Q$, because $p \leq P(\sigma)$ and $p \leq Q(\sigma)$ is equivalent to $p \leq \min\{P(\sigma), Q(\sigma)\}$. 
Consider also quantitative existential quantification: in QFUA triples, $p \geq \Inf x.P(\sigma)$ guarantees $\exists v.\ p \geq P(\sigma[x \mapsto v])$, because the infimum becomes the minimum in discrete domains. The QBUA triples case follows a similar pattern.
Similar duality extends to other operators (see \cref{fig:correspondence}).
Lastly, the Boolean predicate conversion $[B]$ bridges Boolean and quantitative reasoning. For QFUA triples, $[\neg B]$ assigns $-\infty$ when $B$ is true (because $p \geq -\infty$ always holds), and $+\infty$ otherwise. For QBUA triples, $[B]$ assigns $+\infty$ when $B$ is true (because $p \leq +\infty$ always holds), and $-\infty$ otherwise.

\begin{figure}[t!]
\centering
\begin{mathpar}\footnotesize
  \inferrule[(\textdagger:Skip)]{}{\ajudge{P}{\Skip}{P}}
  \hva\and
  \Blue{\inferrule[(F:Assign)]{}{\fjudge{P}{\Assign{x}{e}}{\Inf x'.P[x'/x]\curlyvee[x\ne e[x'/x]]}}}
  \hva\and
  \Red{\inferrule[(B:Assign)]{}{\bjudge{P}{\Assign{x}{e}}{\Sup x'.P[x'/x]\curlywedge[x=e[x'/x]]}}}
 \hva\and
  \Blue{\inferrule[(F:Assume)]{}{\fjudge{P\curlyvee[\neg B]}{\Assume{B}}{P\curlyvee[\neg B]}}}
  \hva\and
  \Red{\inferrule[(B:Assume)]{}{\bjudge{P\curlywedge[B]}{\Assume{B}}{P\curlywedge[B]}}}
  \hva\and
  \inferrule[(\textdagger:Tick)]{}{\ajudge{P}{\Tick{e}}{P-e}}
  \hva\and
  \inferrule[(\textdagger:Seq)]{\ajudge{P}{C_1}{R} \\ \ajudge{R}{C_2}{Q}}{\ajudge{P}{\Seq{C_1}{C_2}}{Q}}
  \hva\and
  \inferrule[(\textdagger:ChoiceL)]{\ajudge{P}{C_1}{Q}}{\ajudge{P}{\Choice{C_1}{C_2}}{Q}}
  \hva\and
  \inferrule[(\textdagger:ChoiceR)]{\ajudge{P}{C_2}{Q}}{\ajudge{P}{\Choice{C_1}{C_2}}{Q}}
  \hva\and
  \inferrule[(\textdagger:Loop)]{\forall n<k.\ajudge{P(n)}{C}{P(n+1)}}{\ajudge{P(0)}{\Loop{C}}{P(k)}}
  \hva\and
  \Blue{\inferrule[(F:Local)]{\fjudge{P}{C}{Q}}{\fjudge{\Inf x.P}{\Local{x}{C}}{\Inf x.Q}}}
  \hva\and
  \Red{\inferrule[(B:Local)]{\bjudge{P}{C}{Q}}{\bjudge{\Sup x.P}{\Local{x}{C}}{\Sup x.Q}}}
  \hva\and
  \Blue{\inferrule[(F:Disj)]{\forall i\in I.\fjudge{P_i}{C}{Q_i}}{\fjudge{\bigcurlywedge_{i\in I}P_i}{C}{\bigcurlywedge_{i\in I}Q_i}}}
  \hva\and
  \Red{\inferrule[(B:Disj)]{\forall i\in I.\bjudge{P_i}{C}{Q_i}}{\bjudge{\bigcurlyvee_{i\in I}P_i}{C}{\bigcurlyvee_{i\in I}Q_i}}}
  \hva\and
  \Blue{\inferrule[(F:Constancy)]{\fjudge{P}{C}{Q}\\\Fv(B)\cap\Mod(C)=\emptyset}{\fjudge{P\curlyvee[B]}{C}{Q\curlyvee[B]}}}
  \hva\and
  \Red{\inferrule[(B:Constancy)]{\bjudge{P}{C}{Q}\\\Fv(B)\cap\Mod(C)=\emptyset}{\bjudge{P\curlywedge[B]}{C}{Q\curlywedge[B]}}}
  \hva\and
  \inferrule[(\textdagger:Relax)]{\ajudge{P}{C}{Q}\\\Fv(F)\cap\Mod(C)=\emptyset}{\ajudge{P+F}{C}{Q+F}}
  \hva\and
  \inferrule[(\textdagger:Cons)]{P\preceq P'\\\ajudge{P'}{C}{Q'}\\Q'\preceq Q}{\ajudge{P}{C}{Q}}
  \hva\and
  \inferrule[(\textdagger:Subst)]{\ajudge{P}{C}{Q}\\y\notin\Fv(P)\cup\Fv(Q)\cup\Fv(C)}{\ajudge{P[y/x]}{C[y/x]}{Q[y/x]}}
\end{mathpar}
\caption{Proof Rules for QFUA and QBUA Triples}
\label{fig:rules}
\end{figure}

We present the proof rules for QFUA and QBUA triples in \cref{fig:rules}. Rules marked with \textdagger\ apply to both QFUA and QBUA triples, while those in \Blue{blue} are specific to QFUA triples, and those in \Red{red} are specific to QBUA triples. Except for the \textsc{(\textdagger:Tick)} and the \textsc{(\textdagger:Relax)} rule, all other rules are generalized from the corresponding Boolean cases, with the operators replaced by their quantitative counterparts.

Concretely, the \textsc{(\textdagger:Skip)} rule trivially preserves quantitative assertions.
The \textsc{(F:Assign)} and \textsc{(B:Assign)} rules for assignment statements are generalized from the standard Floyd assignment rule $\ajudge{P}{\Assign{x}{e}}{\exists x'.P[x'/x] \land x=e[x'/x]}$, which is sound for both QFUA and QBUA triples.
The \textsc{(F:Assume)} and \textsc{(B:Assume)} rules for assume statements are generalized from the assume rule $\ajudge{P\land B}{\Assume{B}}{P\land B}$, which filters out states that do not satisfy the predicate $B$.
The \textsc{(\textdagger:Tick)} rule models the effect of a tick statement by subtracting the tick cost from the resource function.

The \textsc{(\textdagger:Seq)} rule models the sequential composition of two programs, where the postcondition of the first program coincides with the precondition of the second program.
The \textsc{(\textdagger:ChoiceL)} and \textsc{(\textdagger:ChoiceR)} rules approximate non-deterministic choice by considering the behavior of one of the branches.
The \textsc{(\textdagger:Loop)} rule uses an indexed resource function $P(n)$ as a subvariant to approximate loop behavior: if $C$ transforms $P(n)$ to $P(n+1)$ for all $n < k$, then $P(0)$ can transform to $P(k)$ via $\Loop{C}$.
The \textsc{(F:Local)} and \textsc{(B:Local)} rules use the $\Inf$ and $\Sup$ quantifiers, respectively, to erase the effect of local variables, which generalizes the existential quantifier in the Boolean case.

The \textsc{(F:Disj)} and \textsc{(B:Disj)} rules merge multiple triples into a single one by the quantitative corresponding of the disjunction operator.
The \textsc{(F:Constancy)} and \textsc{(B:Constancy)} rule preserves unmodified Boolean constraints — when program $C$ doesn't modify $B$'s free variables, both pre- and post- predicates can be simultaneously conjuncted with $B$ through the quantitative operators.
The \textsc{(\textdagger:Relax)} rule allows additive adjustments to resource functions: if $F$'s variables remain unchanged by $C$, both assertion bounds can be shifted by $F$ through arithmetic addition.
The (\textdagger:Cons) rule shares identical syntax for both QFUA and QBUA. However, the implication direction is different: QFUA rules allow weakening of preconditions and strengthening of postconditions, while QBUA rules allow the opposite.

We next show that the QBUA and QFUA variants are sound and complete with respect to the corresponding semantics.

\begin{theorem}[Soundness of QFUA and QBUA triples]\label{thm:QUA-soundness}
  For all $P,Q,C$:
  \begin{itemize}
    \item If $\fjudge{P}{C}{Q}$ is derivable, then $\fJudge{P}{C}{Q}$ holds.
    \item If $\bjudge{P}{C}{Q}$ is derivable, then $\bJudge{P}{C}{Q}$ holds.
  \end{itemize}
\end{theorem}
\begin{proof}
  By induction on the derivation of the triple. See the Appendix (\cref{thm:qfua-sound,thm:qbua-sound}).
\end{proof}

\begin{theorem}[Completeness of QFUA and QBUA triples]\label{thm:QUA-completeness}
  For all $P,Q,C$ such that $P$ and $Q$ are finitely supported:
  \begin{itemize}
    \item If $\fJudge{P}{C}{Q}$ holds, then $\fjudge{P}{C}{Q}$ is derivable.
    \item If $\bJudge{P}{C}{Q}$ holds, then $\bjudge{P}{C}{Q}$ is derivable.
  \end{itemize}
\end{theorem}
\begin{proof}
  By induction on the structure of $C$. See the Appendix (\cref{thm:qfua-complete,thm:qbua-complete}).
\end{proof}

\subsection{Under-approximating High-water Marks}
\label{sec:technical:hwm}

\begin{figure}[t!]
\centering
\begin{mathpar}\footnotesize
  \inferrule[(\Bd:Skip)]{P\preceq 0}{\djudge{P}{\Skip}{P}}
  \quad 
  \inferrule[(\Bd:Assign)]{P\preceq 0}{\djudge{P}{\Assign{x}{e}}{\Sup x'.P[x'/x]\curlywedge[x=e[x'/x]]}}
  \quad
  \inferrule[(\Bd:Assume)]{P\preceq 0}{\djudge{P\curlywedge[B]}{\Assume{B}}{P\curlywedge[B]}}
  \hva\and
  \inferrule[(\Bd:Tick)]{P\curlywedge P-e\preceq 0}{\djudge{P}{\Tick{e}}{P-e}}
  \quad
  \inferrule[(\Bd:SeqL)]{\djudge{P}{C_1}{R} \\ \bjudge{R}{C_2}{Q}}{\djudge{P}{\Seq{C_1}{C_2}}{Q}}
  \quad
  \inferrule[(\Bd:SeqR)]{\bjudge{P}{C_1}{R} \\ \djudge{R}{C_2}{Q}}{\djudge{P}{\Seq{C_1}{C_2}}{Q}}
  \hva\and
  \inferrule[(\Bd:LoopZero)]{P\preceq 0}{\djudge{P}{C^\star}{P}}
  \quad
  \inferrule[(\Bd:Loop)]{\forall n<k.\bjudge{P(n)}{C}{P(n+1)} \\ \exists m<k.\djudge{P(m)}{C}{P(m+1)}}{\djudge{P(0)}{C^\star}{P(k)}}
  \\
  \inferrule[(\Bd:Disj)]{\forall i\in I.\djudge{P_i}{C}{Q_i}}{\djudge{\bigcurlyvee_{i\in I}P_i}{C}{\bigcurlyvee_{i\in I}Q_i}}
  \hva\and
  \inferrule[(\Bd:Relax)]{\djudge{P}{C}{Q}\\\Fv(F)\cap\Mod(C)=\emptyset\\F\preceq 0}{\djudge{P+F}{C}{Q+F}}
  \hva\and
\end{mathpar}
\caption{Selected Proof Rules for \QBUAd Triples}
\label{fig:QBUAd-rules}
\end{figure}

To under-approximate the worst-case high-water mark,  we refine the QBUA semantics to track minimal resource levels during execution. We introduce \QBUAd triples, whose semantics are defined as follows.

\begin{definition}
  A \QBUAd triple $\dJudge{P}{C}{Q}$ holds if and only if for all $\sigma$ and $p$ such that $p\le P(\sigma)$, there exists a $\tau$ and a $q$ such that $q\le Q(\tau)$ and $\bigstepl{C}{\sigma}{p}{\tau}{q}$.
  This implies that for all $\sigma$ such that $P(\sigma)\in\mathbb{Z}$, there exists an execution of $C$ starting from $\sigma$ that costs at least $P(\sigma)-\sup_{\tau}\{Q(\tau)\}$ ticks, and has a high-water mark of at least $P(\sigma)$.
\end{definition}

Before presenting the proof rules, we explain why the forward variant is not defined.
A natural idea is to refine the QFUA semantics by requiring that the execution must pass through a point where the resource becomes non-positive.
Formally, one could define $\vDash_{\mathsf{F}}^\Diamond\left[P\right]C\left[Q\right]$ to mean that for all $\tau$ and $q \ge Q(\tau)$, there exists $\sigma$ and $p \ge P(\sigma)$ such that $\bigstepl{C}{\sigma}{p}{\tau}{q}$.
However, this definition leads to unintuitive results.
Consider the program $C = \Tick{10};\Tick{(-5)}$.
Intuitively, we would expect $\vDash_{\mathsf{F}}^\Diamond\left[10\right]C\left[5\right]$ to hold.
But if we take $q = 100 \ge 5$, there is no $p$ such that executing $C$ from $p$
leads to $q$ while the potential ever drops to~$\le 0$,
so the triple fails to hold.
This fragility makes the forward definition hard to reason about, whereas the backward formulation (\QBUAd) remains robust and compositional.
For this reason, we only develop the backward variant for reasoning about
high-water marks.

Selected proof rules for \QBUAd triples are shown in \cref{fig:QBUAd-rules}, which are similar to the QBUA rules but with additional constraints.
The full version can be found in the Appendix (\cref{fig:fullqbuad}).
The \textsc{(\Bd:Skip)}, \textsc{(\Bd:Assign)}, and \textsc{(\Bd:Assume)} rules assume that the resource function $P$ is non-positive.
The \textsc{(\Bd:Tick)} rule requires $P\curlywedge P-e\preceq 0$, which ensures that either before or after the tick operation, the resource count is non-positive.
For the sequential composition, if an execution of $\Seq{C_1}{C_2}$ makes the resource counter non-positive at some point, then the point will occur either during the execution of $C_1$ or $C_2$. This is captured by the \textsc{(\Bd:SeqL)} and \textsc{(\Bd:SeqR)} rules.
The loop is similar: if at some point during the execution of $\Loop{C}$ the resource is non-positive, then either $\Loop{C}$ degenerates into a $\Skip$ and the resource is non-positive from the start, or this point will appear during the execution of $C$ at some step. The former case is captured by the \textsc{(\Bd:LoopZero)} rule, while the latter is captured by the subvariant-base \textsc{(\Bd:Loop)} rule, with assuming $\exists m<k.\djudge{P(m)}{C}{P(m+1)}$.
The \textsc{(\Bd:Disj)} rule is similar to the \textsc{(B:Disj)} rule.
Finally, the \textsc{(\Bd:Relax)} rule allows only non-positive additive adjustments to resource functions, which ensures that the point of non-positivity resource is preserved.

\QBUAd is also sound and complete, as shown in the following theorems.

\begin{theorem}[Soundness of \QBUAd triples]
  For all $P,Q,C$, if $\djudge{P}{C}{Q}$ is derivable, then $\dJudge{P}{C}{Q}$ holds.
\end{theorem}

\begin{proof}
  By induction on the derivation of the triple. Details are in the Appendix (\cref{thm:qbuad-sound}).
\end{proof}

\begin{theorem}[Completeness of \QBUAd triples]
  For all $P,Q,C$ such that $P$ and $Q$ are finitely supported, if $\dJudge{P}{C}{Q}$ holds, then $\djudge{P}{C}{Q}$ is derivable.
\end{theorem}

\begin{proof}
  By induction on the structure of $C$. Details are in the Appendix (\cref{thm:qbuad-complete}).
\end{proof}

\section{Implementation}
\label{sec:implementation}

To assess the practicality of our quantitative under-approximate program logic, we implemented a prototype verifier in OCaml targeting a C-style imperative language.
The verifier supports semi-automated checking of the three variants of under-approximate logic introduced in this paper: QFUA, QBUA, and \QBUAd.
Users provide the preconditions, postconditions, and---when needed---loop subvariants (and the index $m$ for high-water mark analysis).
The verifier then checks the validity of the triple using predicate transformers and SMT solving.

\subsection{Quantitative Predicate Transformers}
\label{sec:implementation:transformers}

To support the verification of quantitative under-approximate triples, we adopt and adapt Zhang and Kaminski~\cite{OOPSLA:ZK22}'s quantitative strongest postconditions to define three predicate transformers, each corresponding to one of our three logics: QFUA, QBUA, and \QBUAd.
These transformers generalize classical Boolean predicate transformers---such as strongest postconditions and weakest preconditions---by lifting them to operate over resource functions $State\to\mathbb{Z}\cup\{\pm\infty\}$.
For a loop-free program $C$, all three predicate transformers can be computed compositionally by structural recursion on the syntax of $C$.
The rules for computing the transformers for QFUA and QBUA are summarized in \cref{fig:transformers}, while those for \QBUAd can be found in the Appendix (\cref{fig:transformers-qbuad}).
Below we describe each transformer’s definition, intuition, and role in verification.

\paragraph{QFUA: Strongest Postcondition.}
For QFUA triples, which assert that some execution from a sufficiently resourced pre-state leads to a post-state with a certain resource level, we define a quantitative strongest postcondition transformer:
\begin{align*}
  \post{C}{P}\coloneqq\lambda\tau.\inf\{q:\exists\sigma,p.p\ge P(\sigma)\land\bigstepp{C}{\sigma}{p}{\tau}{q}\}.
\end{align*}
This function maps each post-state $\tau$ to the minimum amount of remaining resource $q$ that can result from running $C$ from some pre-state $\sigma$ with resource at least $P(\sigma)$.
Intuitively, $\post{C}{P}$ gives the lowest postcondition such that the triple $\fjudge{P}{C}{Q}$ is valid, as long as $Q$ lies above it.
To verify a triple $\fjudge{P}{C}{Q}$, it therefore suffices to compute $\post{C}{P}$ and check that $Q\succeq\post{C}{P}$.

\paragraph{QBUA: Weakest Precondition.}
For QBUA triples, which assert that every pre-state with a bounded resource can lead to some post-state where the remaining resource does not exceed a given bound, we define a quantitative weakest precondition transformer:
\begin{align*}
\pre{C}{Q}\coloneqq\lambda\sigma.\sup\{p:\exists\tau,q.q\le Q(\tau)\land\bigstepp{C}{\sigma}{p}{\tau}{q}\}.
\end{align*}
This function maps each pre-state $\sigma$ to the maximum amount of initial resource $p$ such that there exists some execution of $C$ reaching a post-state $\tau$ with remaining resource no greater than $Q(\tau)$.
Intuitively, $\pre{C}{Q}$ gives the greatest precondition such that the triple $\bjudge{P}{C}{Q}$ is valid, as long as $P$ lies above it.
To verify a triple $\bjudge{P}{C}{Q}$, it therefore suffices to compute $\pre{C}{Q}$ and check that $P\preceq\pre{C}{Q}$.

\paragraph{\QBUAd: Weakest Precondition with High-water Mark.}
For \QBUAd triples, which additionally require that the execution passes through a point where the resource drops to zero or below, we define a modified weakest precondition transformer that tracks such high-water behavior:
\begin{align*}
\prel{C}{Q}\coloneqq\lambda\sigma.\sup\{p:\exists\tau,q.q\le Q(\tau)\land\bigstepl{C}{\sigma}{p}{\tau}{q}\}.
\end{align*}
Compared to the transformer for QBUA, this definition adds the side condition that the execution must deplete the resource to at most zero at some intermediate point. The rest of the structure remains unchanged.
This function maps each pre-state $\sigma$ to the maximum amount of initial resource $p$ such that there exists an execution of $C$ that both reaches a post-state $\tau$ bounded by $Q(\tau)$ and passes through a resource level of $\le 0$ during execution.
Similarly to the QBUA case, $\prel{C}{Q}$ gives the greatest precondition such that the triple $\djudge{P}{C}{Q}$ is valid, as long as $P$ lies above it.
And to verify a triple $\djudge{P}{C}{Q}$, it suffices to compute $\prel{C}{Q}$ and check that $P\preceq\prel{C}{Q}$.

\begin{table}[t!]
\centering
\caption{Rules for Computing Quantitative Predicate Transformers $\mathrm{sp}$ and $\mathrm{wp}$}
\label{fig:transformers}

\vspace{0.5em}

\begin{tabular}{l@{\hspace{0.5em}}|@{\hspace{0.5em}}l@{\hspace{0.5em}}|@{\hspace{0.5em}}l}
\toprule
&$\post{C}{P}$&$\pre{C}{Q}$\\
\midrule
$\Skip$&$P$&$Q$\\
$\Assign{x}{e}$&$\Inf x'.P[x'/x]\curlyvee[x\ne e[x'/x]]$&$Q[e/x]$\\
$\Assume{B}$&$P\curlyvee [\neg B]$&$P\curlywedge [B]$\\
$\Seq{C_1}{C_2}$&$\post{C_2}{\post{C_1}{P}}$&$\pre{C_1}{\pre{C_2}{Q}}$\\
$\Choice{C_1}{C_2}$&$\post{C_1}{P}\curlywedge\post{C_2}{P}$&$\pre{C_1}{Q}\curlyvee\pre{C_2}{Q}$\\
$\Local{x}{C}$&$\Inf x.\post{C}{P}$&$\Sup x.\pre{C}{Q}$\\
\bottomrule
\end{tabular}

%
\end{table}

\subsection{The Verification Framework}
\label{sec:implementation:framework}

To implement the program logic described in the previous sections, we design a verification framework that adapts the theoretical model to a more practical setting.
In particular, while our logic formalizes assertions as resource functions of type $State \rightarrow \mathbb{Z} \cup \{\pm\infty\}$, our implementation factorizes these into two components: a Boolean specification predicate that identifies admissible program states, and an integer-valued function that represents the required or remaining resource.
Furthermore, to support realistic program features, our implementation extends the core language with array operations, treating arrays as total functions and encoding read-over-write semantics via axioms compatible with standard SMT array theories.
In the rest of this section, we explain how the predicate transformers introduced earlier are instantiated in this factored setting, how program annotations are provided, and how compositional verification is performed for both loop-free and loop-containing programs.

We represent resource functions $P : State \to \mathbb{Z} \cup \{\pm\infty\}$ in a factored form as pairs $[P_S; P_R]$, where $P_S : State \to \{\true,\false\}$ is a logical specification predicate indicating which states are admissible, and $P_R : State \to \mathbb{Z}$ is a numeric function specifying the resource bound.
The intention is that the resource constraint $P_R$ applies only when $P_S$ holds.
Formally, the pair $[P_S; P_R]$ denotes a resource function whose interpretation depends on the logic.
In the case of QFUA, it corresponds to the function $[\neg P_S] \curlyvee P_R$, which evaluates to $+\infty$ on any state not satisfying $P_S$, and to $P_R$ otherwise.
As a result, only states satisfying $P_S$ are allowed, and on those, the available resource must be at least $P_R$.
In contrast, for QBUA and \QBUAd triples, the pair corresponds to $[P_S] \curlywedge P_R$, which evaluates to $-\infty$ outside $P_S$, again disallowing states that violate the specification, and requires the resource to be at most $P_R$ where $P_S$ holds.
With this notation, we provide some derived rules in \cref{fig:derived-rules}.
The \textdagger\ rule can be used for both QFUA and QBUA triples, while the \textdaggerdbl\ rule is valid for all three kinds of triples.

\begin{figure}[t!]
\begin{mathpar}\footnotesize
  \inferrule[(\textdagger:Assign')]{x\notin\Fv(P_R)}{\AAjudge{P_S}{P_R}{\Assign{x}{e}}{\exists x'.P_S[x'/x]\land x=e[x'/x]}{P_R}}
  \quad
  \inferrule[(\textdaggerdbl:IfTrue)]{\AAjudge{P_S\land B}{P_R}{C_1}{Q_S}{Q_R}}{\AAjudge{P_S}{P_R}{\Ite{B}{C_1}{C_2}}{Q_S}{Q_R}}
  \hva\and
  \inferrule[(\textdagger:WhileSubvar)]{\forall n<k-1.\aajudge{P_S(n)\land B}{P_R(n)}{C}{P_S(n+1)\land B}{P_R(n+1)} \\ \aajudge{P_S(k-1)\land B}{P_R(k-1)}{C}{P_S(k)\land\neg B]}{P_S(k)}}{\aajudge{P_S(0)\land B}{P_R(0)}{\While{B}{C}}{P_S(k)\land\neg B}{P_R(k)}}
\end{mathpar}
\caption{Selected Derived Rules}
\label{fig:derived-rules}
\end{figure}

To verify a program against a quantitative triple, our framework requires the user to annotate the program at both the entry and exit points with pre- and post-conditions of the form $[P_S; P_R]$ described above.

For programs without loops, verification is fully automatic.
Given the annotated pre- and post-conditions and the desired logic (QFUA, QBUA, or \QBUAd), the system performs symbolic predicate transformation and generates the corresponding verification condition.
This condition typically takes the form of an implication between symbolic arithmetic expressions over state variables.
It is then passed to the \textsc{Z3} SMT solver~\cite{TACAS:Z308} to decide validity.

For programs with loops, the user must provide additional annotations.
Specifically, the user supplies a variable $k$ that tracks the number of loop iterations on each execution path, and is required to remain unchanged within the loop body and a loop subvariant expressed as a family of assertions $[P_S(n); P_R(n)]$, indexed by the iteration counter $n$.

For QFUA and QBUA logic, the verifier then checks that for all $n$ such that $0 \le n < k$, the loop body transforms $[P_S(n); P_R(n)]$ into $[P_S(n+1); P_R(n+1)]$, using the logic-specific predicate transformer.
If this inductive check succeeds, the loop is summarized by the precondition $[P_S(0) \land 0 \le k; P_R(0)]$ and the postcondition $[P_S(k) \land 0 \le k; P_R(k)]$, which are used for verifying the continuation of the program.

The \QBUAd logic extends QBUA by requiring that the execution path passes through a point where the resource drops to zero or below.
To capture this, the user must additionally provide an index $m$, indicating the iteration during which this drop is expected to occur.
As with $k$, the value of $m$ must remain unchanged inside the loop.
The verifier first checks, as in QBUA, that for all $0 \le n < k$, the loop body transforms $[P_S(n); P_R(n)]$ into $[P_S(n+1); P_R(n+1)]$.
Then, it performs a second check under \QBUAd, verifying that the loop body transforms $[P_S(m); P_R(m)]$ into $[P_S(m+1); P_R(m+1)]$ via an execution in which the resource drops to $\le 0$ at some point.
Together, these two conditions ensure that the loop behaves correctly both in terms of normal execution and the high-water resource constraint.
The loop is again summarized using $[P_S(0) \land 0 \le m < k; P_R(0)]$ and $[P_S(k) \land 0 \le m < k; P_R(k)]$.

Overall, the user supplies structural annotations at key points---preconditions, postconditions, and loop specifications---and the framework handles the rest.
Predicate transformation reduces the verification problem to logical and arithmetic conditions, which are then discharged using SMT solving.
Our system is fully automatic for straight-line code, and requires only local, bounded annotations for loops to enable modular, compositional verification.

Finally, our verifier supports reasoning over arrays by encoding them as total functions $a : \mathbb{Z} \to \mathbb{Z}$, equipped with two standard operations.
The expression $a[e]$ denotes the value stored at index $e$, while $a\{e \mapsto e'\}$ represents a new array identical to $a$, except that index $e$ is updated to value $e'$.
These operations satisfy the usual read-over-write axioms:
\begin{align*}
  &i=k\to(a\{i \mapsto v\})[k] = v &
  &i\ne k\to(a\{i \mapsto v\})[k] = a[k]
\end{align*}
This abstraction aligns with the standard theory of arrays supported by SMT solvers such as \textsc{Z3}~\cite{TACAS:Z308}.

\section{Case Studies}
\label{sec:case-studies}

In this section, we present a total of eight verified programs as our case studies: four small illustrative combinations of pre- and post-conditions for the example from \cref{sec:overview:qual} and four larger case studies.
All the case studies are successfully verified using our prototype implementation, and the corresponding annotated code is provided in the supplementary material.

We begin this section by revisiting the example from \cref{sec:overview:qual} to demonstrate how our prototype verifier operates concretely on annotated programs.
Consider the following simple conditional:

\medskip

\begin{algorithmic}[1]\small\setstretch{0.75}
  \State \Gray{\texttt{/*}~precondition: $P_1=[\top;~2]$~\texttt{*/} \texttt{/*}~precondition: $P_2=[x=42;~2]$~\texttt{*/}}
  \State \texttt{if (x == 42) \{ tick(2); x = 0; \} else \{ tick(1); \}}
  \State \Gray{\texttt{/*}~postcondition: $Q_1=[\top;~0]$~\texttt{*/} \texttt{/*}~postcondition: $Q_2=[x=0;~0]$~\texttt{*/}}
\end{algorithmic}

\medskip

Let us test the four combinations of these pre- and post-conditions under both QFUA and QBUA.
The pair $(P_1, Q_1)$ is not valid under either logic.
The pair $(P_1, Q_2)$ is only valid under QFUA, and $(P_2, Q_1)$ is only valid under QBUA.
Finally, $(P_2, Q_2)$ is valid under both.
These results match the analysis given in \cref{sec:overview:qual}, where we discussed how QFUA and QBUA emphasize output- and input-oriented reasoning, respectively.

\subsection{QFUA: Password Validation}

This example illustrates the use of QFUA triples to characterize states that can potentially incur high resource consumption.
We consider a password validation program where the password is an array of $n$ integers, and the user must input $n$ integers to match it entry-by-entry.
Each user input is modeled by a \texttt{tick(1)} command, incurring one unit of resource.
The program stops early if a mismatch is found, but if the password is fully correct, it consumes $n$ units of input effort.
QFUA reasoning allows us to formally express that in any final state where the variable \texttt{valid} equals $1$, the execution might have consumed up to $n$ units of resource.

\medskip

\begin{algorithmic}[1]\small\setstretch{0.75}
  \State \Gray{\texttt{/*}~precondition: $[n\ge 0;~n]$\texttt{*/}}
  \State \texttt{int valid = 1;}
  \State \texttt{int i = 0;}
  \State \texttt{while (i < n \&\& valid == 1)}
  \State \Gray{\texttt{/*}~number of iterations: $n$~\texttt{;}}
  \State \Gray{\texttt{ \ }~subvariant: $i_0 \mapsto [
    i = i_0 \land valid = 1;~
    n - i_0]$~\texttt{*/}}
  \State \texttt{\{}
  \Indent
  \State \texttt{tick(1);} \Comment{user input}
  \State \texttt{int input;}
  \State \texttt{if (input != password[i])}
  \Indent
  \State \texttt{valid = 0;}
  \EndIndent
  \State \texttt{i = i + 1;}
  \EndIndent
  \State \texttt{\}}
  \State \Gray{\texttt{/*}~postcondition: $[n \ge 0 \land valid = 1;~0]$~\texttt{*/}}
\end{algorithmic}

\medskip

This QFUA-annotated program expresses that reaching a final state where $valid = 1$---i.e., successful password verification---may cost up to $n$ units of input, corresponding to the longest possible execution path.
The loop subvariant uses the index $i_0$ to represent progress through the password, and the quantitative postcondition tracks how many inputs remain to be matched.
By computing the strongest postcondition, the verifier confirms that only successful verification can justify the worst-case input cost.

While the same triple holds under QBUA, it does not capture that the high resource usage occurs only when $valid = 1$. In fact, even if we drop this constraint from the postcondition, the QBUA triple remains valid---unlike QFUA, which becomes invalid. This highlights QFUA’s expressiveness in specifying output-sensitive resource behavior.

\subsection{QBUA: Insertion Sort}

This example uses QBUA triples to under-approximate the total number of swaps performed by insertion sort on a worst-case input.
Compared to QFUA, QBUA is particularly suitable here because we are interested in identifying initial states (i.e., specific array configurations) that lead to executions incurring high resource usage.

The program below shows the insertion sort algorithm, and it is annotated with loop subvariants and resource assertions.
The variable \texttt{n} denotes the length of the array \texttt{a}, and \texttt{tick(1)} represents one unit of resource consumption per swap.
This annotation expresses that for input arrays sorted in strictly decreasing order, the total number of swaps executed by insertion sort is under-approximated by $n(n-1)/2$, which is tight.

\medskip

\begin{algorithmic}[1]\small\setstretch{0.75}
  \State \Gray{\texttt{/*}~precondition: $[n\ge 1\land\forall I\in[0,n).~\forall J\in[0,I).~a[J]>a[I];~n(n-1)/2]$~\texttt{*/}}
  \State \texttt{int i = 1;}
  \State \texttt{while (i < n)}
  \State \Gray{\texttt{/*}~number of iterations: $n-1$\texttt{;}}
  \State \Gray{\texttt{ \ }~subvariant: $i_0\mapsto\begin{bmatrix}i=i_0+1\land\forall I\in[i,n).~\forall J\in[0,I).~a[J]>a[I];\\-i_0(i_0+1)/2\end{bmatrix}$\texttt{;}}
  \State \Gray{\texttt{ \ }~constant prefix: $[\top;~n(n-1)/2]$~\texttt{*/}}
  \State \texttt{\{}
  \Indent
  \State \texttt{int j = i;}
  \State \texttt{while (j > 0)}
  \State \Gray{\texttt{/*}~number of iterations: $i$\texttt{;}}
  \State \Gray{\texttt{ \ }~subvariant: $j_0\mapsto\begin{bmatrix}j=i-j_0\land\forall I\in[0,j).~a[I]>a[j]\\\land\forall I\in[i+1,n).~\forall J\in[0,I).~a[J]>a[I];\\-j_0\end{bmatrix}$\texttt{;}}
  \State \Gray{\texttt{ \ }~constant prefix: $[i=i_0+1;~-i_0(i_0+1)/2]$~\texttt{*/}}
  \State \texttt{\{}
  \Indent
  \State \texttt{if (a[j - 1] > a[j]) \{}
  \Indent
  \State \texttt{tick(1);} \Comment{swap}
  \State \texttt{int t = a[j];}
  \State \texttt{a[j] = a[j - 1];}
  \State \texttt{a[j - 1] = t;}
  \EndIndent
  \State \texttt{\}}
  \State \texttt{j = j - 1;}
  \EndIndent
  \State \texttt{\}}
  \State \texttt{i = i + 1;}
  \EndIndent
  \State \texttt{\}}
  \State \Gray{\texttt{/*}~postcondition: $[n\ge 1;~0]$~\texttt{*/}}
\end{algorithmic}

\medskip

To verify this fact, our system introduces index variables $i_0$ and $j_0$ to define subvariants that track the outer and inner loop states, respectively.
The equalities $i = i_0 + 1$ and $j = i - j_0$ make express the relationship between the program variables \texttt{i}, \texttt{j} and the logical iteration counters $i_0$, $j_0$ that parameterize the subvariant families.

The logical part of these assertions captures key ordering properties of the array.
In the outer loop, the subvariant $\forall I\in[i,n).~\forall J\in[0,I).~a[J] > a[I]$ ensures that the suffix of the array from index $i$ onward remains sorted in strictly decreasing order.
This reflects that the portion of the array not yet visited by the outer loop is still in its worst-case configuration.
In the inner loop, the subvariant $\forall I\in[0,j).~a[I] > a[j]$ asserts that the element currently being inserted (at index $j$) is strictly smaller than all elements to its left, capturing the progress of the inner loop as it ``bubbles'' the element leftward.
In addition, the condition $\forall I\in[i+1,n).~\forall J\in[0,I).~a[J] > a[I]$ preserves the decreasing order in the untouched suffix beyond the current outer loop index, maintaining the worst-case assumption throughout execution.

The numeric part of each subvariant quantifies the resource consumption within the corresponding loop.
For the outer loop, this is expressed by the term $-i_0(i_0+1)/2$, which tracks the cumulative number of swaps performed after $i_0$ iterations.
For the inner loop, the numeric subvariant $-j_0$ reflects the swaps required for inserting the current element.
Together, these expressions describe the local resource evolution of each loop, enabling modular reasoning about their quantitative behavior.

To complete the reasoning for the entire program, each loop is further associated with a constant prefix that remains unchanged during the loop execution but contributes to the overall resource analysis.
For the outer loop, the constant prefix is applied after deriving the loop triple, using the \textsc{(B:Relax)} rule to uniformly shift both the pre- and post-potentials by $n(n-1)/2$ , thereby aligning the result with the precondition of the entire program.
For the inner loop, the constant prefix restores the context of the outer loop once the inner loop reasoning is finished: the logical condition $i=i_0+1$ through the \textsc{(B:Constancy)} rule, and the resource offset $-i_0(i_0+1)/2$ through the \textsc{(B:Relax)} rule.
These constant prefixes involve only variables that remain unchanged by the loop body and thus need not be reasoned about within the loop itself, yet they ensure that the verified subvariants compose correctly into the global quantitative reasoning of the program.

\subsection{QBUA: Quicksort with Stack}

We also apply QBUA reasoning to analyze the worst-case number of comparisons performed by quicksort.
To simplify reasoning, we implement quicksort using an explicit stack to simulate recursion, storing subarray bounds in \texttt{lstk} and \texttt{rstk}.
Instead of actually choosing pivots randomly, we use the command \texttt{assume(pivot >= l \&\& pivot < r)} to model nondeterministic pivot selection.

Our goal is to verify that when the input array is initially sorted in strictly increasing order, there exists a pivot selection strategy under which quicksort performs at least $n(n-1)/2$ comparisons.
This lower bound is tight and corresponds to the degenerate case where each pivot is chosen to be the last element in the subarray.
We model each comparison by inserting a \texttt{tick(1)} statement before any data movement occurs.

\medskip

\begin{algorithmic}[1]\small\setstretch{0.75}
  \State \Gray{\texttt{/*}~precondition: $[n \ge 1 \land \forall I \in [0,n).~\forall J \in [0,I).~a[J] < a[I];~n(n-1)/2]$~\texttt{*/}}
  \State \texttt{int top = 0;}
  \State \texttt{lstk[top] = 0;}
  \State \texttt{rstk[top] = n;}
  \State \texttt{top = top + 1;}
  \State \texttt{while (top > 0)}
  \State \Gray{\texttt{/*}~number of iterations: $n$\texttt{;}}
  \State \Gray{\texttt{ \ }~subvariant: $t \mapsto \begin{bmatrix}
    \text{if } t < n: &
    top = 1 \land lstk[0] = 0 \land rstk[0] = n - t \\
      & \land \forall I \in [0, n - t).~\forall J \in [0,I).~a[J] < a[I] \\
    \text{else:} &
      top = 0; \\
    & - \frac{t(2n - t - 1)}{2}
    \end{bmatrix}$}
  \State \Gray{\texttt{ \ }~constant prefix: $[n\ge 1;~n(n-1)/2]$~\texttt{*/}}
  \State \texttt{\{}
  \Indent
  \State \texttt{top = top - 1;}
  \State \texttt{int l = lstk[top];}
  \State \texttt{int r = rstk[top];}
  \State \texttt{int pivot;}
  \State \texttt{assume(pivot >= l \&\& pivot < r);}
  \State \texttt{int i = l;}
  \State \texttt{int j = l;}
  \State \texttt{int k = r;}
  \State \texttt{while (i < r)}
  \State \Gray{\texttt{/*}~number of iterations: $r - l$\texttt{;}}
  \State \Gray{\texttt{ \ }~subvariant: $i_0 \mapsto \begin{bmatrix}
    & pivot = r - 1 \land i = l + i_0 \\
    \land & j =
      \begin{cases}
        i & \text{if } i \le pivot \\
        i - 1 & \text{otherwise}
      \end{cases}
    \land k = r \\
    \land & \forall I \in [l,r).~\forall J \in [l,I).~a[J] < a[I] \\
    \land & \forall I \in [l,j).~b[I] =
      \begin{cases}
        a[I] & \text{if } I < pivot \\
        a[I + 1] & \text{otherwise}
      \end{cases}; \\
          & -
      \begin{cases}
        i_0 & \text{if } i \le pivot \\
        i_0 - 1 & \text{otherwise}
      \end{cases}
    \end{bmatrix}$}
  \State \Gray{\texttt{ \ }~constant prefix: $[top = 0 \land l = 0 \land r = n - t;~-t(2n - t - 1)/2]$~\texttt{*/}}
  \State \texttt{\{}
  \Indent
  \State \texttt{if (i != pivot) \{}
  \Indent
  \State \texttt{tick(1);} \Comment{comparison}
  \State \texttt{if (a[i] <= a[pivot]) \{}
  \Indent
  \State \texttt{b[j] = a[i];}
  \State \texttt{j = j + 1;}
  \EndIndent
  \State \texttt{\} else \{}
  \Indent
  \State \texttt{k = k - 1;}
  \State \texttt{b[k] = a[i];}
  \EndIndent
  \State \texttt{\}}
  \EndIndent
  \State \texttt{\}}
  \State \texttt{i = i + 1;}
  \EndIndent
  \State \texttt{\}}
  \State \texttt{b[j] = a[pivot];}
  \State \texttt{i = l;}
  \State \texttt{while (i < r)}
  \State \Gray{\texttt{/*}~number of iterations: $r - l$\texttt{;}}
  \State \Gray{\texttt{ \ }~subvariant: $i_0 \mapsto [i = l + i_0\land \forall I \in [l,i).~a[I] = b[I];~0]$}
  \State \Gray{\texttt{ \ }~constant prefix: $\begin{bmatrix}&top = 0 \land l = 0 \land r = n - t \land j = r - 1 \land k = r \\\land &\forall I \in [l,j).~\forall J \in [l,I).~b[J] < b[I];\\&t(2n - t - 1)/2 - (n - t - 1)\end{bmatrix}$~\texttt{*/}}
  \State \texttt{\{}
  \Indent
  \State \texttt{a[i] = b[i];}
  \State \texttt{i = i + 1;}
  \EndIndent
  \State \texttt{\}}
  \State \texttt{if (l < j) \{}
  \Indent
  \State \texttt{lstk[top] = l;}
  \State \texttt{rstk[top] = j;}
  \State \texttt{top = top + 1;}
  \EndIndent
  \State \texttt{\}}
  \State \texttt{if (j + 1 < r) \{}
  \Indent
  \State \texttt{lstk[top] = j + 1;}
  \State \texttt{rstk[top] = r;}
  \State \texttt{top = top + 1;}
  \EndIndent
  \State \texttt{\}}
  \EndIndent
  \State \texttt{\}}
  \State \Gray{\texttt{/*}~postcondition: $[n\ge 1;~0]$~\texttt{*/}}
\end{algorithmic}

\medskip

To verify this bound, we introduce a subvariant index $t$ to count the number of partitioning steps performed by the outer loop.
The subvariant maintains that the first $n - t$ elements of the array remain sorted, and that the stack contains at most one active subarray to process.

In the inner partitioning loop, a second subvariant indexed by $i_0$ expresses the progress of scanning through the current subarray $[l, r)$, along with logical assertions capturing the ordering of values and the correctness of the constructed temporary array $b$.
The second inner loop copies values from $b$ back into $a$, while maintaining that the written prefix is sorted.

The resource subvariants compute the total number of comparisons remaining, and allow us to verify that the number of \texttt{tick(1)} operations accumulates to at least $n(n-1)/2$ along a path where each pivot is selected as the last element of its subarray.

\subsection{\QBUAd: Producer-consumer System}

We now present an example that uses \QBUAd to verify a lower bound on resource consumption along some execution path---specifically, that the buffer may reach full capacity, meaning the resource is completely exhausted.

The program models a producer-consumer system with nondeterministic scheduling.
There are initially $n$ production steps and $n$ consumption steps, and at each iteration, the environment nondeterministically chooses to either produce or consume.
This nondeterminism is encoded using a special Boolean variable \texttt{demon}, where the conditional \texttt{if (demon)} indicates a nondeterministic choice between branches.
Each production step adds one item to the buffer and is modeled by \texttt{tick(1)}, which consumes one unit of resource---that is, one unit of available buffer space.
Each consumption step removes one item and is modeled by \texttt{tick(-1)}, which frees up buffer space, increasing the resource.

The variable \texttt{buf} tracks the current buffer occupancy, while \texttt{p} and \texttt{c} represent the number of remaining production and consumption steps, respectively.
In this setting, the resource represents the remaining buffer capacity. When the buffer becomes full, the resource value reaches zero. Our goal is to show that, under some execution strategy, the buffer occupancy may reach $n$ at some point---that is, the total resource consumption reaches at least $n$ before the program terminates.

\medskip

\begin{algorithmic}[1]\small\setstretch{0.75}
  \State \Gray{\texttt{/*}~precondition: $[n \ge 0;~n]$~\texttt{*/}}
  \State \texttt{int p = n;}
  \State \texttt{int c = n;}
  \State \texttt{int buf = 0;}
  \State \texttt{while (p + c > 0)}
  \State \Gray{\texttt{/*}~number of iterations: $2n$\texttt{;}}
  \State \Gray{\texttt{ \ }~subvariant: $i \mapsto \begin{bmatrix}
      \begin{cases}
      buf = i \land p = n - i \land c = n & \text{if } i < n \\
      buf = 2n - i \land p = 0 \land c = 2n - i & \text{otherwise}
      \end{cases};\\
      \begin{cases}
        n - i & \text{if } i < n \\
        i - n & \text{otherwise}
    \end{cases}
  \end{bmatrix}$\texttt{;}}
  \State \Gray{\texttt{ \ }exhaustion point: $n$~\texttt{*/}}
  \State \texttt{\{}
  \Indent
  \State \texttt{if (demon) \{}
  \Indent
  \State \texttt{assume(p > 0);}
  \State \texttt{p = p - 1;}
  \State \texttt{tick(1);} \Comment{addition to buffer}
  \State \texttt{buf = buf + 1;}
  \EndIndent
  \State \texttt{\} else \{}
  \Indent
  \State \texttt{assume(buf > 0 \&\& c > 0);}
  \State \texttt{c = c - 1;}
  \State \texttt{tick(-1);} \Comment{removal from buffer}
  \State \texttt{buf = buf - 1;}
  \EndIndent
  \State \texttt{\}}
  \EndIndent
  \State \texttt{\}}
  \State \Gray{\texttt{/*}~postcondition: $[n\ge 0;~n]$~\texttt{*/}}
\end{algorithmic}

\medskip

The subvariant for this loop is indexed by $i$, representing the total number of steps taken so far.
The first $n$ steps correspond to pure production, during which the buffer usage increases and reaches its peak value $n$ at step $i = n$.
The remaining $n$ steps simulate pure consumption, reducing the buffer back to zero.
In addition to tracking the logical evolution of variables---$p$ decreasing and $buf$ increasing before $i = n$, and $c$ decreasing and $buf$ decreasing afterward---the subvariant also includes an explicit exhaustion point at $i = n$.
This marks the iteration at which the resource value (i.e., remaining buffer capacity) reaches zero, indicating the buffer is full.

The \QBUAd verifier confirms there exists an execution path where the resource value eventually drops to zero---i.e., the buffer becomes full---and that execution can continue to termination without breaking the postcondition.

\section{Related Work}
\label{sec:related}


\paragraph*{Incorrectness Logic}
Incorrectness logic (IL), introduced by O'Hearn~\cite{POPL:OHearn20}, signifies a major shift in program verification by emphasizing \emph{under-approximate} reasoning instead of the \emph{over-approximate} approach employed in classic Hoare logic.
The core idea of IL is the use of \emph{forward, under-approximate} (FUA) triples to ensure that any bug detected is a true positive.
FUA triples were initially studied by de Vries and Koutavas~\cite{SEFM:dVK11}, and they discussed \emph{backward, under-approximate} (BUA) triples as \emph{total} Hoare triples.
BUA triples were later studied by M{\"o}ller et al.~\cite{RAMiCS:MHH21} and Ascari et al.~\cite{arxiv:ABG23}.
As a subroutine of \textsc{UNTer} for under-approximating non-termination, Raad et al.~\cite{OOPSLA:RVO24} studied the theory of BUA triples and realized them in practice.
Le et al.~\cite{OOPSLA:LRV22} extends IL with separation logic (ISL) to find bugs in heap-manipulating programs.
Raad et al.~\cite{POPL:RBD22} further developed a concurrent variant of ISL to detect concurrency bugs, e.g., races and deadlocks.
Zilberstein et al.~\cite{OOPSLA:ZDS23,OOPSLA:ZSS24} developed Outcome logic (OL) and Outcome separation logic (OSL) as a unified foundation for reasoning about both correctness and incorrectness.
Among the aforementioned studies, only \textsc{UNTer}~\cite{OOPSLA:RVO24} focuses on the resource usage of programs, as it aims to prove non-termination, i.e., a \emph{qualitative} justification for a program's usage of the time resource.
Recently, Cousot~\cite{POPL:Cousot25} studies transformational program logics for correctness and incorrectness, including termination and non-termination.
In this work, we aim to establish \emph{quantitative} under-approximations of the worst-case resource usage of programs. 

\paragraph*{Resource-aware Program Logic}
Haslbeck and Nipkow~\cite{TACAS:HN18} reviewed three variants of Hoare logic for establishing time bounds.
The first one is the quantitative Hoare logic (QHL)~\cite{PLDI:CHR14,PLDI:CHS15} that we reviewed in \cref{sec:overview:qhl}.
Note that QHL is not restricted to time bounds; it can also handle general resources like stack bounds.
QHL generalizes Boolean assertions by introducing the concept of \emph{potentials} to reason about resource consumption.
In QHL, the difference between the pre- and the post-potential provides an upper bound (i.e., an over-approximation) on the resource consumption.
Furthermore, because potentials can never be negative, the pre-potential also establishes an upper bound on the high-water mark for executing a program.
The second one is a Hoare-like logic to prove big-O-style time bounds, developed by Nielson~\cite{SCP:Nielson87}.
Nielson's logic establishes triples of the form $\left\{p\right\} C \left\{ E \Downarrow q \right\}$, where $p$ and $q$ are Boolean assertions and $E$ is a time bound.
The third one is a line of work of incorporating separation logic with potential-like notions such as \emph{time credits}~\cite{ESOP:Atkey10,ESOP:MJP19,POPL:PGJ24,phd:Gueneau19}.
The high-level idea is to treat resources (or potentials) as a special type of heap fragment, allowing us to use \emph{separating conjunction} to annotate program states with quantitative information, e.g., $\left\{p \ast \$3\right\} C \left\{ q \ast \$1 \right\}$ indicates that the pre-state carries $3$ units of resources and the post-state carries $1$ unit of resources.
All the aforementioned logics focus on over-approximating worst-case resource usage (i.e., sound upper bounds) or under-approximating best-case resource usage (i.e., sound lower bounds).
In this work, we aim to develop a program logic to under-approximate worst-case resource usage.
We could have developed our work based on time credits or similar mechanisms to address the compositionality and modularity issues mentioned in \cref{sec:overview:qual}, but we find it orthogonal to the under-approximation problem we aim to tackle in this paper.
In other words, we still need to devise the forward and backward reasoning principles for resource under-approximation, as well as the treatment of high-water marks.
Nevertheless, extending our work with time credits within the context of separation logic is an interesting area for our future work, especially in the setting of establishing \emph{asymptotic} bounds~\cite{phd:Gueneau19}.

Graded Hoare logic (GHL)~\cite{ESOP:GKO21} augments Hoare logic with a \emph{preordered monoidal analysis} to reason about side-effects, including resource usage, probabilistic behavior, and differential privacy-related properties.
Given a preordered monoid $(M,{\le},1,{\cdot})$, GHL proves judgements of the form $\vdash_m \left\{ \phi \right\} C \left\{ \psi \right\}$, where $m \in M$ is called the \emph{analysis} of program $C$ and $\phi,\psi$ are the pre- and post-conditions.
For example, one can use the natural-number monoid $(\bbN,{\le},0,{+})$ to instantiate GHL to reason about resource usage, e.g., $\vdash_0 \left\{\phi\right\} \Skip \left\{\phi\right\}$ and $\vdash_2 \left\{\phi\right\} \Tick{2} \left\{\phi\right\}$.
GHL still follows the methodology of \emph{over-approximating}.
%

\paragraph*{Worst-case Resource Analysis}
Most existing approaches for identifying reachable worst cases are based on testing, using techniques such as random testing and symbolic execution to generate concrete inputs that can lead to significant resource usage. 
Recent fuzzing-based methods include \textsc{KelinciWCA}~\cite{ISSTA:NKP18}, \textsc{SlowFuzz}~\cite{CCS:PZK17}, and resource-usage-aware fuzzing~\cite{FASE:CHL22}; they are considered black-box methods as they do not require knowledge of the program's content.
On the other hand, symbolic-execution-based methods look into the concrete structures of programs; to name a few, WISE~\cite{ICSE:BJS09}, SPF-WCA~\cite{ICST:LKP17}, \textsc{Badger}~\cite{ISSTA:NKP18}, and resource-type-guided symbolic execution~\cite{POPL:WH19}.
All the aforementioned methods can generate concrete inputs, i.e., they provide \emph{concrete} under-approximations of the worst-case resource usage, whereas our quantitative under-approximate logic proves \emph{symbolic} under-approximations.
%
%
%

\subsubsection*{Data Availability Statement.} The artifact accompanying this paper, including the OCaml implementation of the prototype verifier and the mechanized Rocq proofs, is available at Zenodo~\cite{software:JW26}.

\subsubsection*{Acknowledgements.} This work was sponsored by National Key R\&D Program of China, Grant No. 2024YFE0204100.

\bibliography{db}

\appendix
\newpage

%
%
%




\begin{table}[t!]
\centering
\caption{Rules for Computing Quantitative Predicate Transformer $\mathrm{wp}^{\le 0}$}
\label{fig:transformers-qbuad}

\vspace{0.5em}

 \begin{tabular}{l@{\hspace{0.5em}}|@{\hspace{0.5em}}l@{\hspace{0.5em}}}
 \toprule
 &$\prel{C}{Q}$\\
 \midrule
 $\Skip$&$Q\curlywedge 0$\\
 $\Assign{x}{e}$&$Q[e/x]\curlywedge 0$\\
 $\Assume{B}$&$P\curlywedge [B]\curlywedge 0$\\[8pt]
 $\Seq{C_1}{C_2}$&$\begin{matrix}\prel{C_1}{\pre{C_2}{Q}}\\\curlywedge\pre{C_1}{\prel{C_2}{Q}}\end{matrix}$\\[12pt]
 $\Choice{C_1}{C_2}$&$\prel{C_1}{Q}\curlyvee\prel{C_2}{Q}$\\
 $\Local{x}{C}$&$\Sup x.\prel{C}{Q}$\\
 \bottomrule
 \end{tabular}
\end{table}

\section{Proofs}

All proofs in this Appendix are also mechanized in Rocq.
The corresponding scripts are available in the artifact~\cite{software:JW26}.

\begin{figure}
\begin{mathpar}
  \inferrule[(BS:Skip)]{}{\bigstep{\Skip}{\sigma}{p}{p}{\sigma}{p}}
  \hva\and
  \inferrule[(BS:Assign)]{}{\bigstep{\Assign{x}{e}}{\sigma}{p}{p}{\sigma[x\mapsto\eval{e}{\sigma}]}{p}}
  \hva\and
  \inferrule[(BS:Assume)]{\eval{B}{\sigma}=\true}{\bigstep{\Assume{B}}{\sigma}{p}{p}{\sigma}{p}}
  \hva\and
  \inferrule[(BS:Tick)]{}{\bigstep{\Tick{e}}{\sigma}{p}{\min\{p,p-\eval{e}{\sigma}\}}{\sigma}{p-\eval{e}{\sigma}}}
  \hva\and
  \inferrule[(BS:Seq)]{\bigstep{C_1}{\sigma}{p}{l_1}{\rho}{r} \\ \bigstep{C_2}{\rho}{r}{l_2}{\tau}{q}}{\bigstep{\Seq{C_1}{C_2}}{\sigma}{p}{\min\{l_1,l_2\}}{\tau}{q}}
  \hva\and
  \inferrule[(BS:ChoiceL)]{\bigstep{C_1}{\sigma}{p}{l}{\tau}{q}}{\bigstep{\Choice{C_1}{C_2}}{\sigma}{p}{l}{\tau}{q}}
  \hva\and
  \inferrule[(BS:ChoiceR)]{\bigstep{C_2}{\sigma}{p}{l}{\tau}{q}}{\bigstep{\Choice{C_1}{C_2}}{\sigma}{p}{l}{\tau}{q}}
  \hva\and
  \inferrule[(BS:LoopZero)]{}{\bigstep{\Loop{C}}{\sigma}{p}{p}{\sigma}{p}}
  \hva\and
  \inferrule[(BS:Loop)]{\bigstep{\Seq{C}{\Loop{C}}}{\sigma}{p}{l}{\tau}{q}}{\bigstep{\Loop{C}}{\sigma}{p}{l}{\tau}{q}}
  \hva\and
  \inferrule[(BS:Local)]{\bigstep{C}{\sigma}{p}{l}{\tau}{q}}{\bigstep{\Local{x}{C}}{\sigma[x\mapsto v]}{p}{l}{\tau[x\mapsto v]}{q}}
\end{mathpar}
\caption{Rules for the Big-step Semantics}
\label{fig:fullsemantics}
\end{figure}

\cref{fig:fullsemantics} shows the complete rules for the big-step semantics.
We also write $\bigstepp{C}{\sigma}{p}{\tau}{q}$ to denote $\exists l.\bigstep{C}{\sigma}{p}{l}{\tau}{q}$, and $\bigstepl{C}{\sigma}{p}{\tau}{q}$ to denote $\exists l\le 0.\bigstep{C}{\sigma}{p}{l}{\tau}{q}$.

\begin{lemma}\label{lem:relax}
  For all $C,\sigma,p,\tau,q,l,f$, if $\bigstep{C}{\sigma}{p}{l}{\tau}{q}$, then $\bigstep{C}{\sigma}{p+f}{l+f}{\tau}{q+f}$.
\end{lemma}

\begin{proof}
  By induction on the derivation.
  \begin{itemize}
    \item \textbf{Case} \textsc{(BS:Skip)}, \textsc{(BS:Assign)}, \textsc{(BS:Assume)}:\\
      Immediate by the corresponding rule.
    \item \textbf{Case} \textsc{(BS:Tick)}:\\
      By (BS:Tick), we have $\bigstep{\Tick{e}}{\sigma}{p+f}{\min\{p+f,p-\eval{e}{\sigma}+f\}}{\sigma}{p-\eval{e}{\sigma}+f}$. Since $\min\{p+f,p-\eval{e}{\sigma}+f\}=\min\{p,p-\eval{e}{\sigma}\}+f$, the conclusion holds.
    \item \textbf{Case} \textsc{(BS:Seq)}:\\
      By the induction hypothesis, we have $\bigstep{C_1}{\sigma}{p+f}{l_1+f}{\rho}{r+f}$ and $\bigstep{C_2}{\rho}{r+f}{l_2+f}{\tau}{q+f}$. Then by (BS:Seq), we have $\bigstep{\Seq{C_1}{C_2}}{\sigma}{p+f}{\min\{l_1+f,l_2+f\}}{\tau}{q+f}$. Since $\min\{l_1+f,l_2+f\}=\min\{l_1,l_2\}+f$, the conclusion holds.
    \item \textbf{Case} \textsc{(BS:ChoiceL)}:\\
      By the induction hypothesis, we have $\bigstep{C_1}{\sigma}{p+f}{l+f}{\tau}{q+f}$. Then by (BS:ChoiceL), we have $\bigstep{\Choice{C_1}{C_2}}{\sigma}{p+f}{l+f}{\tau}{q+f}$.
    \item \textbf{Case} \textsc{(BS:ChoiceR)}:\\
      This case is similar to the previous case.
    \item \textbf{Case} \textsc{(BS:LoopZero)}:\\
      Immediate by the rule (BS:LoopZero).
    \item \textbf{Case} \textsc{(BS:Loop)}:\\
      By the induction hypothesis, we have $\bigstep{\Seq{C}{\Loop{C}}}{\sigma}{p+f}{l+f}{\tau}{q+f}$. Then by (BS:Loop), we have $\bigstep{\Loop{C}}{\sigma}{p+f}{l+f}{\tau}{q+f}$.
    \item \textbf{Case} \textsc{(BS:Local)}:\\
      By the induction hypothesis, we have $\bigstep{C}{\sigma}{p+f}{l+f}{\tau}{q+f}$. Then by (BS:Local), we have $\bigstep{\Local{x}{C}}{\sigma[x\mapsto v]}{p+f}{l+f}{\tau[x\mapsto v]}{q+f}$.
  \end{itemize}
\end{proof}

\begin{definition}[Semantics]
  The semantics of QFUA triple, QBUA triple, and \QBUAd triple are defined as follows:
  \begin{itemize}
    \item $\fJudge{P}{C}{Q}$ holds iff for all $\tau$ and $q$ such that $q\ge Q(\tau)$, there exists a $\sigma$ and a $p$ such that $p\ge P(\sigma)$ and $\bigstepp{C}{\sigma}{p}{\tau}{q}$.
    \item $\bJudge{P}{C}{Q}$ holds iff for all $\sigma$ and $p$ such that $p\le P(\sigma)$, there exists a $\tau$ and a $q$ such that $q\le Q(\tau)$ and $\bigstepp{C}{\sigma}{p}{\tau}{q}$.
    \item $\dJudge{P}{C}{Q}$ holds iff for all $\sigma$ and $p$ such that $p\le P(\sigma)$, there exists a $\tau$ and a $q$ such that $q\le Q(\tau)$ and $\bigstepl{C}{\sigma}{p}{\tau}{q}$.
  \end{itemize}
\end{definition}

\begin{definition}[Predicate Transformer] The $\mathrm{pre}$ and $\mathrm{post}$ transformers are defined as follows:
  \begin{itemize}
    \item $\pre{C}{Q}\coloneqq\lambda\sigma.\sup\{p:\exists\tau,q.q\le Q(\tau)\land\bigstepp{C}{\sigma}{p}{\tau}{q}\}$.
    \item $\prel{C}{Q}\coloneqq\lambda\sigma.\sup\{p:\exists\tau,q.q\le Q(\tau)\land\bigstepl{C}{\sigma}{p}{\tau}{q}\}$.
    \item $\post{C}{P}\coloneqq\lambda\tau.\inf\{q:\exists\sigma,p.p\ge P(\sigma)\land\bigstepp{C}{\sigma}{p}{\tau}{q}\}$.
  \end{itemize}
\end{definition}

\begin{proposition}\label{prop:const}
  For all $P,C$ such that $\Fv(P)\cap\Mod(C)=\emptyset$, for all $\sigma,p,\tau,q,l$, if $\bigstep{C}{\sigma}{p}{l}{\tau}{q}$, then $P(\sigma)=P(\tau)$.
\end{proposition}

%

\begin{proposition}\label{prop:subst}
  For all $C,\sigma,p,\tau,q,l$ such that $\bigstep{C}{\sigma}{p}{l}{\tau}{q}$, for all $x,y,v$, if $y\notin\Fv(C)$, then $\bigstep{C[y/x]}{\sigma[x\mapsto v][y\mapsto\eval{x}{\sigma}]}{p}{l}{\tau[x\mapsto v][y\mapsto\eval{x}{\tau}]}{q}$.
\end{proposition}

\begin{figure}
\begin{mathpar}
  \inferrule[(F:Skip)]{}{\fjudge{P}{\Skip}{P}}
  \hva\and
  \inferrule[(F:Assign)]{}{\fjudge{P}{\Assign{x}{e}}{\Inf x'.P[x'/x]\curlyvee[x\ne e[x'/x]]}}
  \hva\and
  \inferrule[(F:Assume)]{}{\fjudge{P\curlyvee[\neg B]}{\Assume{B}}{P\curlyvee[\neg B]}}
  \hva\and
  \inferrule[(F:Tick)]{}{\fjudge{P}{\Tick{e}}{P-e}}
  \hva\and
  \inferrule[(F:Seq)]{\fjudge{P}{C_1}{R} \\ \fjudge{R}{C_2}{Q}}{\fjudge{P}{\Seq{C_1}{C_2}}{Q}}
  \hva\and
  \inferrule[(F:ChoiceL)]{\fjudge{P}{C_1}{Q}}{\fjudge{P}{\Choice{C_1}{C_2}}{Q}}
  \hva\and
  \inferrule[(F:ChoiceR)]{\fjudge{P}{C_2}{Q}}{\fjudge{P}{\Choice{C_1}{C_2}}{Q}}
  \hva\and
  \inferrule[(F:Loop)]{\forall n<k.\fjudge{P(n)}{C}{P(n+1)}}{\fjudge{P(0)}{\Loop{C}}{P(k)}}
  \hva\and
  \inferrule[(F:Local)]{\fjudge{P}{C}{Q}}{\fjudge{\Inf x.P}{\Local{x}{C}}{\Inf x.Q}}
  \hva\and
  \inferrule[(F:Disj)]{\forall i\in I.\fjudge{P_i}{C}{Q_i}}{\fjudge{\bigcurlywedge_{i\in I}P_i}{C}{\bigcurlywedge_{i\in I}Q_i}}
  \hva\and
  \inferrule[(F:Constancy)]{\fjudge{P}{C}{Q}\\\Fv(B)\cap\Mod(C)=\emptyset}{\fjudge{P\curlyvee[B]}{C}{Q\curlyvee[B]}}
  \hva\and
  \inferrule[(F:Relax)]{\fjudge{P}{C}{Q}\\\Fv(F)\cap\Mod(C)=\emptyset}{\fjudge{P+F}{C}{Q+F}}
  \hva\and
  \inferrule[(F:Cons)]{P\preceq P'\\\fjudge{P'}{C}{Q'}\\Q'\preceq Q}{\fjudge{P}{C}{Q}}
  \hva\and
  \inferrule[(F:Subst)]{\fjudge{P}{C}{Q}\\y\notin\Fv(P)\cup\Fv(Q)\cup\Fv(C)}{\fjudge{P[y/x]}{C[y/x]}{Q[y/x]}}
\end{mathpar}
\caption{Proof Rules for QFUA Triples}
\label{fig:fullqfua}
\end{figure}

\cref{fig:fullqfua} shows the complete proof rules for QFUA triples.

\begin{theorem}[QFUA Soundness]\label{thm:qfua-sound}
  For all $P$, $C$, $Q$, if $\fjudge{P}{C}{Q}$ is derivable, then $\fJudge{P}{C}{Q}$ holds.
\end{theorem}

\begin{proof}
  By induction on the derivation.
  \begin{itemize}
    \item \textbf{Case} \textsc{(F:Skip)}:\\
      Let $P$ be arbitrary such that $\fjudge{P}{\Skip}{P}$. For any $\tau$ and $q$ satisfying $q\ge P(\tau)$, take $\sigma=\tau$ and $p=q$. Then we have $p\ge P(\sigma)$ and $\bigstep{\Skip}{\sigma}{p}{q}{\tau}{q}$ by \textsc{(BS:Skip)}.
    \item \textbf{Case} \textsc{(F:Assign)}:\\
      Let $P$ be arbitrary such that $\fjudge{P}{\Assign{x}{e}}{Q}$, where $Q=\Inf x'.P[x'/x]\curlyvee[x\ne e[x'/x]]$. For any $\tau$ and $q$ satisfying $q\ge Q(\tau)$, notice that
      \begin{align*}
        Q(\tau)
        &=\inf_v\left\{(P[x'/x]\curlyvee[x\ne e[x'/x]])(\tau[x'\mapsto v])\right\}\\
        &=\inf_v\{\max\{P[x'/x](\tau[x'\mapsto v]),[x\ne e[x'/x]](\tau[x'\mapsto v])\}\}\\
        &=\inf_v\{\max\{P(\tau[x\mapsto v]),[\eval{x}{\tau}\ne\eval{e}{\tau[x\mapsto v]}]\}\}.
      \end{align*}
      Thus, there exists a $v$ such that $q\ge P(\tau[x\mapsto v])$ and $\eval{x}{\tau}=\eval{e}{\tau[x\mapsto v]}$. Take $\sigma=\tau[x\mapsto v]$ and $p=q$, then we have $p\ge P(\sigma)$ and $\tau=\sigma[x\mapsto\eval{e}{\sigma}]$. Therefore, $\bigstep{\Assign{x}{e}}{\sigma}{p}{q}{\sigma}{q}$ holds by \textsc{(BS:Assign)}.
    \item \textbf{Case} \textsc{(F:Assume)}:\\
      Let $P$ be arbitrary such that $\fjudge{P\curlyvee[\neg B]}{\Assume{B}}{P\curlyvee[\neg B]}$. For any $\tau$ and $q$ satisfying $q\ge(P\curlyvee[\neg B])(\tau)$, we have $q\ge P(\tau)$ and $\eval{B}{\tau}=\true$. Take $\sigma=\tau$ and $p=q$. Then we have $p\ge(P\curlyvee[\neg B])(\sigma)$ and $\bigstep{\Assume{B}}{\sigma}{p}{q}{\sigma}{q}$ by \textsc{(BS:Assume)}.
    \item \textbf{Case} \textsc{(F:Tick)}:\\
      Let $P$ be arbitrary such that $\fjudge{P}{\Tick{e}}{P-e}$. For any $\tau$ and $q$ satisfying $q\ge (P-e)(\tau)=P(\tau)-\eval{e}{\tau}$, take $\sigma=\tau$ and $p=q+\eval{e}{\tau}$. Then we have $p\ge P(\sigma)$ and $\bigstep{\Tick{e}}{\sigma}{p}{\min\{p,p-\eval{e}{\sigma}\}}{\tau}{q}$ by \textsc{(BS:Tick)}.
    \item \textbf{Case} \textsc{(F:Seq)}:\\
      Let $P,Q,R,C_1,C_2$ be arbitrary such that $\fjudge{P}{C_1}{R}$ and $\fjudge{R}{C_2}{Q}$. By induction hypothesis, we have $\fJudge{P}{C_1}{R}$ and $\fJudge{R}{C_2}{Q}$. For any $\tau$ and $q$ satisfying $q\ge Q(\tau)$, there exists a $\rho$ and an $r$ such that $r\ge R(\rho)$ and $\bigstep{C_2}{\rho}{r}{l_2}{\tau}{q}$ for some $l_2$. Then, there exists a $\sigma$ and a $p$ such that $p\ge P(\sigma)$ and $\bigstep{C_1}{\sigma}{p}{l_1}{\rho}{r}$ for some $l_1$. Hence, $\bigstep{\Seq{C_1}{C_2}}{\sigma}{p}{\min\{l_1,l_2\}}{\tau}{q}$ holds by \textsc{(BS:Seq)}.
    \item \textbf{Case} \textsc{(F:ChoiceL)}:\\
      Let $P,Q,C_1,C_2$ be arbitrary such that $\fjudge{P}{C_1}{Q}$. By induction hypothesis, we have $\fJudge{P}{C_1}{Q}$. For any $\tau$ and $q$ satisfying $q\ge Q(\tau)$, there exists a $\sigma$ and a $p$ such that $p\ge P(\sigma)$ and $\bigstep{C_1}{\sigma}{p}{l}{\tau}{q}$ for some $l$. Hence, $\bigstep{\Choice{C_1}{C_2}}{\sigma}{p}{l}{\tau}{q}$ holds by \textsc{(BS:ChoiceL)}.
    \item \textbf{Case} \textsc{(F:ChoiceR)}:\\
      This case is similar to the previous one.
    \item \textbf{Case} \textsc{(F:Loop)}:\\
      Let $P,C$ be arbitrary such that $\forall n<k.\fjudge{P(n)}{C}{P(n+1)}$. By induction hypothesis, we have $\forall n<k.\fJudge{P(n)}{C}{P(n+1)}$. We will inductively prove that for all $0\le i\le k$, $\fJudge{P(k-i)}{\Loop{C}}{P(k)}$ holds:
      \begin{itemize}
        \item \textbf{Base case} $i=0$:\\
          Let $\tau$ and $q$ be arbitrary such that $q\ge P(k)(\tau)$. Take $\sigma=\tau$ and $p=q$. Then we have $p\ge P(k)(\sigma)$ and $\bigstep{\Loop{C}}{\sigma}{p}{p}{\sigma}{p}$ by \textsc{(BS:LoopZero)}. It follows that $\fJudge{P(k)}{\Loop{C}}{P(k)}$.
        \item \textbf{Inductive step} $i\to i+1$:\\
          Assume that $\fJudge{P(k-i)}{\Loop{C}}{P(k)}$ holds. Let $\tau$ and $q$ be arbitrary such that $q\ge P(k)(\tau)$. By the assumption, there exists a $\rho$ and an $r$ such that $r\ge P(k-i)(\rho)$ and $\bigstep{\Loop{C}}{\rho}{r}{l_1}{\tau}{q}$ for some $l_1$. Besides, $\fJudge{P(k-(i+1))}{C}{P(k-i)}$ implies that there exists a $\sigma$ and a $p$ such that $p\ge P(k-(i+1))(\sigma)$ and $\bigstep{C}{\sigma}{p}{l_2}{\rho}{r}$ for some $l_2$. Thus, $\bigstep{\Seq{C}{\Loop{C}}}{\sigma}{p}{l}{\tau}{q}$ holds by \textsc{(BS:Seq)} for $l=\min\{l_1,l_2\}$, and then $\bigstep{\Loop{C}}{\sigma}{p}{l}{\tau}{q}$ holds by \textsc{(BS:Loop)}. It follows that $\fJudge{P(k-(i+1))}{\Loop{C}}{P(k)}$.
      \end{itemize}
      Taking $i=k$, we have $\fJudge{P(0)}{\Loop{C}}{P(k)}$.
    \item \textbf{Case} \textsc{(F:Local)}:\\
      Let $P,Q,C$ be arbitrary such that $\fjudge{P}{C}{Q}$. By induction hypothesis, we have $\fJudge{P}{C}{Q}$. For any $\tau$ and $q$ satisfying $q\ge (\Inf x.Q)(\tau)=\inf_v\{Q(\tau[x\mapsto v])\}$, there exists a $v$ such that $q\ge Q(\tau[x\mapsto v])$. Since $\fJudge{P}{C}{Q}$, there exists a $\sigma$ and a $p$ such that $p\ge P(\sigma)$ and $\bigstep{C}{\sigma}{p}{l}{\tau[x\mapsto v]}{q}$ for some $l$. This implies that $p\ge (\Inf x.P)(\sigma[x\mapsto\eval{x}{\tau}])$ and $\bigstep{\Local{x}{C}}{\sigma[x\mapsto\eval{x}{\tau}]}{p}{l}{\tau}{q}$ by \textsc{(BS:Local)}. It follows that $\fJudge{\Inf x.P}{\Local{x}{C}}{\Inf x.Q}$.
    \item \textbf{Case} \textsc{(F:Disj)}:\\
      Let $C$ be arbitrary. And for $i\in I$, let $P_i$ and $Q_i$ be arbitrary such that $\fjudge{P_i}{C}{Q_i}$. By induction hypothesis, we have $\fJudge{P_i}{C}{Q_i}$. For any $\tau$ and $q$ satisfying $q\ge\left(\bigcurlywedge_{i\in I}Q_i\right)(\tau)$, there exists an $i\in I$ such that $q\ge Q_i(\tau)$. Since $\fJudge{P_i}{C}{Q_i}$, there exists a $\sigma$ and a $p$ such that $p\ge P_i(\sigma)$ and $\bigstep{C}{\sigma}{p}{l}{\tau}{q}$ for some $l$. Therefore, we have $p\ge\left(\bigcurlywedge_{i\in I}P_i\right)(\sigma)$ and $\bigstep{C}{\sigma}{p}{l}{\tau}{q}$.
    \item \textbf{Case} \textsc{(F:Constancy)}:\\
      Let $P,Q,B,C$ be arbitrary such that $\fjudge{P}{C}{Q}$ and $\Fv(B)\cap\Mod(C)=\emptyset$. By induction hypothesis, we have $\fJudge{P}{C}{Q}$. For any $\tau$ and $q$ satisfying $q\ge (Q\curlyvee[B])(\tau)$, we have $q\ge Q(\tau)$ and $\eval{\neg B}{\tau}=\true$. Then there exists a $\sigma$ and a $p$ such that $p\ge P(\sigma)$ and $\bigstep{C}{\sigma}{p}{l}{\tau}{q}$ for some $l$. By proposition \ref{prop:const}, we have $\eval{\neg B}{\sigma}=\true$, and then $p\ge (P\curlyvee[B])(\sigma)$. It follows that $\fJudge{P\curlyvee[B]}{C}{Q\curlyvee[B]}$.
    \item \textbf{Case} \textsc{(F:Relax)}:\\
      Let $P,Q,F,C$ be arbitrary such that $\fjudge{P}{C}{Q}$ and $\Fv(F)\cap\Mod(C)=\emptyset$. By induction hypothesis, we have $\fJudge{P}{C}{Q}$. For any $\tau$ and $q$ satisfying $q\ge (Q+F)(\tau)=Q(\tau)+F(\tau)$, we have $q-F(\tau)\ge Q(\tau)$. Then there exists a $\sigma$ and a $p$ such that $p\ge P(\sigma)$ and $\bigstep{C}{\sigma}{p}{l}{\tau}{q-F(\tau)}$ for some $l$. By lemma \ref{lem:relax} and proposition \ref{prop:const}, we have $\bigstep{C}{\sigma}{p+F(\sigma)}{l+F(\sigma)}{\tau}{q}$. Since $p+F(\sigma)\ge P(\sigma)+F(\sigma)=(P+F)(\sigma)$, we have $\fJudge{P+F}{C}{Q+F}$.
    \item \textbf{Case} \textsc{(F:Cons)}:\\
      Let $P,P',Q,Q',C$ be arbitrary such that $P\preceq P'$, $Q'\preceq Q$ and $\fjudge{P'}{C}{Q'}$. By induction hypothesis, we have $\fJudge{P'}{C}{Q'}$. For any $\tau$ and $q$ satisfying $q\ge Q(\tau)$, we have $q\ge Q'(\tau)$. Then there exists a $\sigma$ and a $p$ such that $p\ge P'(\sigma)$ and $\bigstep{C}{\sigma}{p}{l}{\tau}{q}$ for some $l$. Since $P(\sigma)\le P'(\sigma)$, we have $p\ge P(\sigma)$. It follows that $\fJudge{P}{C}{Q}$.
    \item \textbf{Case} \textsc{(F:Subst)}:\\
      Let $P,Q,C$ be arbitrary such that $\fjudge{P}{C}{Q}$, and let $x,y$ by arbitrary such that $y\notin\Fv(P)\cup\Fv(Q)\cup\Fv(C)$. By induction hypothesis, we have $\fJudge{P}{C}{Q}$. For any $\tau$ and $q$ satisfying $q\ge Q[y/x](\tau)=Q(\tau[x\mapsto\eval{y}{\tau}])$, we have $q\ge Q(\tau')$ since $y\notin\Fv(Q)$, where $\tau'=\tau[x\mapsto\eval{y}{\tau}][y\mapsto\eval{x}{\tau}]$. Then there exists a $\sigma$ and a $p$ such that $p\ge P(\sigma)$ and $\bigstep{C}{\sigma}{p}{l}{\tau'}{q}$ for some $l$. By proposition \ref{prop:subst}, we have $\bigstep{C[y/x]}{\sigma'}{p}{l}{\tau}{q}$, where $\sigma'=\sigma[x\mapsto\eval{x}{\tau}][y\mapsto\eval{x}{\sigma}]$. Since $y\notin\Fv(P)$, weh have $p\ge P(\sigma)=P[y/x](\sigma')$. It follows that $\fJudge{P[y/x]}{C[y/x]}{Q[y/x]}$.
  \end{itemize}
\end{proof}

\begin{figure}
\begin{mathpar}
  \inferrule[(B:Skip)]{}{\bjudge{P}{\Skip}{P}}
  \hva\and
  \inferrule[(B:Assign)]{}{\bjudge{P}{\Assign{x}{e}}{\Sup x'.P[x'/x]\curlywedge[x=e[x'/x]]}}
  \hva\and
  \inferrule[(B:Assume)]{}{\bjudge{P\curlywedge[B]}{\Assume{B}}{P\curlywedge[B]}}
  \hva\and
  \inferrule[(B:Tick)]{}{\bjudge{P}{\Tick{e}}{P-e}}
  \hva\and
  \inferrule[(B:Seq)]{\bjudge{P}{C_1}{R} \\ \bjudge{R}{C_2}{Q}}{\bjudge{P}{\Seq{C_1}{C_2}}{Q}}
  \hva\and
  \inferrule[(B:ChoiceL)]{\bjudge{P}{C_1}{Q}}{\bjudge{P}{C_1+C_2}{Q}}
  \hva\and
  \inferrule[(B:ChoiceR)]{\bjudge{P}{C_2}{Q}}{\bjudge{P}{C_1+C_2}{Q}}
  \hva\and
  \inferrule[(B:Loop)]{\forall n<k.\bjudge{P(n)}{C}{P(n+1)}}{\bjudge{P(0)}{C^\star}{P(k)}}
  \hva\and
  \inferrule[(B:Local)]{\bjudge{P}{C}{Q}}{\bjudge{\Sup x.P}{\Local{x}{C}}{\Sup x.Q}}
  \hva\and
  \inferrule[(B:Disj)]{\forall i\in I.\bjudge{P_i}{C}{Q_i}}{\bjudge{\bigcurlyvee_{i\in I}P_i}{C}{\bigcurlyvee_{i\in I}Q_i}}
  \hva\and
  \inferrule[(B:Constancy)]{\bjudge{P}{C}{Q}\\\Fv(B)\cap\Mod(C)=\emptyset}{\bjudge{P\curlywedge[B]}{C}{Q\curlywedge[B]}}
  \hva\and
  \inferrule[(B:Relax)]{\bjudge{P}{C}{Q}\\\Fv(F)\cap\Mod(C)=\emptyset}{\bjudge{P+F}{C}{Q+F}}
  \hva\and
  \inferrule[(B:Cons)]{P\preceq P' \\ \judge{P'}{C}{Q'} \\ Q'\preceq Q}{\judge{P}{C}{Q}}
  \hva\and
  \inferrule[(B:Subst)]{\bjudge{P}{C}{Q}\\y\notin\Fv(P)\cup\Fv(Q)\cup\Fv(C)}{\bjudge{P[y/x]}{C[y/x]}{Q[y/x]}}
\end{mathpar}
\caption{Proof Rules for QBUA Triples}
\label{fig:fullqbua}
\end{figure}

\cref{fig:fullqbua} shows the complete proof rules for QBUA triples.

\begin{theorem}[QBUA Soundness]\label{thm:qbua-sound}
  For all $P$, $C$, $Q$, if $\bjudge{P}{C}{Q}$ is derivable, then $\bJudge{P}{C}{Q}$ holds.
\end{theorem}

\begin{proof}
  By induction on the derivation.
  \begin{itemize}
    \item \textbf{Case} \textsc{(B:Skip)}:\\
      Let $P$ be arbitrary such that $\bjudge{P}{\Skip}{P}$. For any $\sigma$ and $p$ satisfying $p\le P(\sigma)$, take $\tau=\sigma$ and $q=p$. Then we have $q\le P(\tau)$ and $\bigstep{\Skip}{\sigma}{p}{p}{\tau}{q}$ by \textsc{(BS:Skip)}.
    \item \textbf{Case} \textsc{(B:Assign)}:\\
      Let $P$ be arbitrary such that $\bjudge{P}{\Assign{x}{e}}{Q}$, where $Q=\Sup x'.P[x'/x]\curlywedge[x=e[x'/x]]$. For any $\sigma$ and $p$ satisfying $p\le P(\sigma)$, take $\tau=\sigma[x\mapsto\eval{e}{\sigma}]$ and $q=p$. Notice that
      \begin{align*}
        Q(\tau)
        &=\sup_v\{(P[x'/x]\curlywedge[x=e[x'/x]])(\tau[x\mapsto v])\}\\
        &=\sup_v\{\min\{(P[x'/x])(\tau[x\mapsto v]),[x=e[x'/x]](\tau[x\mapsto v])\}\}\\
        &=\sup_v\{\min\{P(\tau[x\mapsto v]),[\eval{x}{\tau}=\eval{e}{\tau[x\mapsto v]}]\}\}\\
        &\ge\min\{P(\tau[x\mapsto\eval{x}\sigma]),[\eval{x}{\tau}=\eval{e}{\tau[x\mapsto\eval{\sigma}{x}]}]\}\\
        &=P(\sigma).
    \end{align*}
    Combining with $\bigstep{\Assign{x}{e}}{\sigma}{p}{p}{\tau}{q}$ by \textsc{(BS:Assign)}, we have $\bJudge{P}{\Assign{x}{e}}{Q}$.
    \item \textbf{Case} \textsc{(B:Assume)}:\\
      Let $P$ be arbitrary such that $\bjudge{P\curlywedge[B]}{\Assume{B}}{P\curlywedge[B]}$. For any $\sigma$ and $p$ satisfying $p\le (P\curlywedge[B])(\sigma)$, we have $p\le P(\sigma)$ and $\eval{B}{\sigma}=\true$. Take $\tau=\sigma$ and $q=p$. Then we have $q\le (P\curlywedge[B])(\tau)$ and $\bigstep{\Assume{B}}{\sigma}{p}{p}{\tau}{q}$ by \textsc{(BS:Assume)}.
    \item \textbf{Case} \textsc{(B:Tick)}:\\
      Let $P$ be arbitrary such that $\bjudge{P}{\Tick{e}}{P-e}$. For any $\sigma$ and $p$ satisfying $p\le P(\sigma)$, take $\tau=\sigma$ and $q=p-\eval{e}{\sigma}$. Then we have $q\le (P-e)(\tau)$ and $\bigstep{\Tick{e}}{\sigma}{p}{\min\{p,p-\eval{e}{\sigma}\}}{\tau}{q}$ by \textsc{(BS:Tick)}.
    \item \textbf{Case} \textsc{(B:Seq)}:\\
      Let $P,Q,R,C_1,C_2$ be arbitrary such that $\bjudge{P}{C_1}{R}$ and $\bjudge{R}{C_2}{Q}$. By induction hypothesis, we have $\bJudge{P}{C_1}{R}$ and $\bJudge{R}{C_2}{Q}$. For any $\sigma$ and $p$ satisfying $p\le P(\sigma)$, there exists a $\rho$ and an $r$ such that $r\le R(\rho)$ and $\bigstep{C_1}{\sigma}{p}{l_1}{\rho}{r}$ for some $l_1$. Then, there exists a $\tau$ and a $q$ such that $q\le Q(\tau)$ and $\bigstep{C_2}{\rho}{r}{l_2}{\tau}{q}$ for some $l_2$. Hince, $\bigstep{\Seq{C_1}{C_2}}{\sigma}{p}{\min\{l_1,l_2\}}{\tau}{q}$ holds by \textsc{(BS:Seq)}.
    \item \textbf{Case} \textsc{(B:ChoiceL)}:\\
      Let $P,Q,C_1,C_2$ be arbitrary such that $\bjudge{P}{C_1}{Q}$. By induction hypothesis, we have $\bJudge{P}{C_1}{Q}$. For any $\sigma$ and $p$ satisfying $p\le P(\sigma)$, there exists a $\tau$ and a $q$ such that $q\le Q(\tau)$ and $\bigstep{C_1}{\sigma}{p}{l}{\tau}{q}$ for some $l$. Hence, $\bigstep{\Choice{C_1}{C_2}}{\sigma}{p}{l}{\tau}{q}$ holds by \textsc{(BS:ChoiceL)}.
    \item \textbf{Case} \textsc{(B:ChoiceR)}:\\
      This case is similar to the previous one.
    \item \textbf{Case} \textsc{(B:Loop)}:\\
      Let $P,C$ be arbitrary such that $\forall n<k.\bjudge{P(n)}{C}{P(n+1)}$. By induction hypothesis, we have $\forall n<k.\bJudge{P(n)}{C}{P(n+1)}$. We will inductively prove that for all $0\le i\le k$, $\bJudge{P(k-i)}{\Loop{C}}{P(k)}$ holds:
      \begin{itemize}
        \item \textbf{Base case} $i=0$:\\
          Let $\sigma$ and $p$ be arbitrary such that $p\le P(k)(\sigma)$. Take $\tau=\sigma$ and $q=p$. Then we have $q\le P(k)(\tau)$ and $\bigstep{\Loop{C}}{\sigma}{p}{p}{\tau}{q}$ by \textsc{(BS:LoopZero)}. It follows that $\bJudge{P(k)}{\Loop{C}}{P(k)}$.
        \item \textbf{Inductive step} $i\to i+1$:\\
          Assume that $\bJudge{P(k-i)}{\Loop{C}}{P(k)}$ holds. Let $\sigma$ and $p$ be arbitrary such that $p\le P(k-(i+1))(\sigma)$. By $\bJudge{P(k-(i+1))}{C}{P(k-i)}$, there exists a $\rho$ and an $r$ such that $r\le P(k-i)(\rho)$ and $\bigstep{C}{\sigma}{p}{l_1}{\rho}{r}$ for some $l_1$. Besides, $\bJudge{P(k-i)}{\Loop{C}}{P(k)}$ implies that there exists a $\tau$ and a $q$ such that $q\le P(k)(\tau)$ and $\bigstep{\Loop{C}}{\rho}{r}{l_2}{\tau}{q}$ for some $l_2$. Thus, $\bigstep{\Seq{C}{\Loop{C}}}{\sigma}{p}{l}{\tau}{q}$ holds by \textsc{(BS:Seq)} for $l=\min\{l_1,l_2\}$, and then $\bigstep{\Loop{C}}{\sigma}{p}{l}{\tau}{q}$ holds by \textsc{(BS:Loop)}. It follows that $\bJudge{P(k-(i+1))}{\Loop{C}}{P(k)}$.
        \end{itemize}
        Taking $i=k$, we have $\bJudge{P(0)}{\Loop{C}}{P(k)}$.
    \item \textbf{Case} \textsc{(B:Local)}:\\
      Let $P,Q,C$ be arbitrary such that $\bjudge{P}{C}{Q}$. By induction hypothesis, we have $\bJudge{P}{C}{Q}$. For any $\sigma$ and $p$ satisfying $p\le (\Sup x.P)(\sigma)=\sup_v\{P(\sigma[x\mapsto v])\}$, there exists a $v$ such that $p\le P(\sigma[x\mapsto v])$. Since $\bJudge{P}{C}{Q}$, there exists a $\tau$ and a $q$ such that $q\le Q(\tau)$ and $\bigstep{C}{\sigma[x\mapsto v]}{p}{l}{\tau}{q}$ for some $l$. This implies that $q\le(\Sup x.Q)(\tau[x\mapsto\eval{x}{\sigma}])$ and $\bigstep{\Local{x}{C}}{\sigma}{p}{l}{\tau[x\mapsto\eval{x}{\sigma}]}{q}$ by \textsc{(BS:Local)}. It follows that $\bJudge{\Sup x.P}{\Local{x}{C}}{\Sup x.Q}$.
    \item \textbf{Case} \textsc{(B:Disj)}:\\
      Let $C$ be arbitrary. And for $i\in I$, let $P_i$ and $Q_i$ be arbitrary such that $\bjudge{P_i}{C}{Q_i}$. By induction hypothesis, we have $\bJudge{P_i}{C}{Q_i}$. For any $\sigma$ and $p$ satisfying $p\le\left(\bigcurlyvee_{i\in I}P_i\right)(\sigma)$, there exists an $i\in I$ such that $p\le P_i(\sigma)$. Since $\bJudge{P_i}{C}{Q_i}$, there exists a $\tau$ and a $q$ such that $q\le Q_i(\tau)$ and $\bigstep{C}{\sigma}{p}{l}{\tau}{q}$ for some $l$. Therefore, we have $q\le\left(\bigcurlyvee_{i\in I}Q_i\right)(\tau)$ and $\bigstep{C}{\sigma}{p}{l}{\tau}{q}$.
    \item \textbf{Case} \textsc{(B:Constancy)}:\\
      Let $P,Q,B,C$ be arbitrary such that $\bjudge{P}{C}{Q}$ and $\Fv(B)\cap\Mod(C)=\emptyset$. By induction hypothesis, we have $\bJudge{P}{C}{Q}$. For any $\sigma$ and $p$ satisfying $p\le (P\curlywedge[B])(\sigma)$, we have $p\le P(\sigma)$ and $\eval{B}{\sigma}=\true$. Then there exists a $\tau$ and a $q$ such that $q\le Q(\tau)$ and $\bigstep{C}{\sigma}{p}{l}{\tau}{q}$ for some $l$. By proposition \ref{prop:const}, we have $\eval{B}{\tau}=\true$, and then $q\le (Q\curlywedge[B])(\tau)$. It follows that $\bJudge{P\curlywedge[B]}{C}{Q\curlywedge[B]}$.
    \item \textbf{Case} \textsc{(B:Relax)}:\\
      Let $P,Q,F,C$ be arbitrary such that $\bjudge{P}{C}{Q}$ and $\Fv(F)\cap\Mod(C)=\emptyset$. By induction hypothesis, we have $\bJudge{P}{C}{Q}$. For any $\sigma$ and $p$ satisfying $p\le (P+F)(\sigma)=P(\sigma)+F(\sigma)$, we have $p-F(\sigma)\le P(\sigma)$. Then there exists a $\tau$ and a $q$ such that $q\le Q(\tau)$ and $\bigstep{C}{\sigma}{p-F(\sigma)}{l}{\tau}{q}$ for some $l$. By lemma \ref{lem:relax} and proposition \ref{prop:const}, we have $\bigstep{C}{\sigma}{p}{l+F(\tau)}{\tau}{q+F(\tau)}$. Since $q+F(\tau)\le Q(\tau)+F(\tau)=(Q+F)(\tau)$, we have $\bJudge{P+F}{C}{Q+F}$.
    \item \textbf{Case} \textsc{(B:Cons)}:\\
      Let $P,P',Q,Q',C$ be arbitrary such that $P\preceq P'$, $Q'\preceq Q$ and $\bjudge{P'}{C}{Q'}$. By induction hypothesis, we have $\bJudge{P'}{C}{Q'}$. For any $\sigma$ and $p$ satisfying $p\le P(\sigma)$, we have $p\le P'(\sigma)$. Then there exists a $\tau$ and a $q$ such that $q\le Q'(\tau)$ and $\bigstep{C}{\sigma}{p}{l}{\tau}{q}$ for some $l$. Since $Q'(\tau)\le Q(\tau)$, we have $q\le Q(\tau)$. It follows that $\bJudge{P}{C}{Q}$.
    \item \textbf{Case} \textsc{(B:Subst)}:\\
      Let $P,Q,C$ be arbitrary such that $\bjudge{P}{C}{Q}$, and let $x,y$ be arbitrary such that $y\notin\Fv(P)\cup\Fv(Q)\cup\Fv(C)$. By induction hypothesis, we have $\bJudge{P}{C}{Q}$. For any $\sigma$ and $p$ satisfying $p\le P[y/x](\sigma)=P(\sigma[x\mapsto\eval{y}{\sigma}])$, we have $p\le P(\sigma')$ since $y\notin\Fv(P)$, where $\sigma'=\sigma[x\mapsto\eval{y}{\sigma}][y\mapsto\eval{x}{\sigma}]$. Then there exists a $\tau$ and a $q$ such that $q\le Q(\tau)$ and $\bigstep{C}{\sigma'}{p}{l}{\tau}{q}$ for some $l$. By proposition \ref{prop:subst}, we have $\bigstep{C[y/x]}{\sigma}{p}{l}{\tau'}{q}$, where $\tau'=\tau[x\mapsto\eval{x}{\sigma}][y\mapsto\eval{x}{\tau}]$. Since $y\notin\Fv(Q)$, we have $q\le Q(\tau)=Q[y/x](\tau')$. It follows that $\bJudge{P[y/x]}{C[y/x]}{Q[y/x]}$.
  \end{itemize}
\end{proof}

\begin{figure}
\begin{mathpar}
  \inferrule[(\Bd:Skip)]{P\preceq 0}{\djudge{P}{\Skip}{P}}
  \hva\and
  \inferrule[(\Bd:Assign)]{P\preceq 0}{\djudge{P}{\Assign{x}{e}}{\Sup x'.P[x'/x]\curlywedge[x=e[x'/x]]}}
  \hva\and
  \inferrule[(\Bd:Assume)]{P\preceq 0}{\djudge{P\curlywedge[B]}{\Assume{B}}{P\curlywedge[B]}}
  \hva\and
  \inferrule[(\Bd:Tick)]{P\curlywedge P-e\preceq 0}{\djudge{P}{\Tick{e}}{P-e}}
  \hva\and
  \inferrule[(\Bd:SeqL)]{\djudge{P}{C_1}{R} \\ \bjudge{R}{C_2}{Q}}{\djudge{P}{\Seq{C_1}{C_2}}{Q}}
  \hva\and
  \inferrule[(\Bd:SeqR)]{\bjudge{P}{C_1}{R} \\ \djudge{R}{C_2}{Q}}{\djudge{P}{\Seq{C_1}{C_2}}{Q}}
  \hva\and
  \inferrule[(\Bd:ChoiceL)]{\djudge{P}{C_1}{Q}}{\djudge{P}{C_1+C_2}{Q}}
  \hva\and
  \inferrule[(\Bd:ChoiceR)]{\djudge{P}{C_2}{Q}}{\djudge{P}{C_1+C_2}{Q}}
  \hva\and
  \inferrule[(\Bd:LoopZero)]{P\preceq 0}{\djudge{P}{C^\star}{P}}
  \hva\and
  \inferrule[(\Bd:Loop)]{\forall n<k.\bjudge{P(n)}{C}{P(n+1)} \\ \exists m<k.\djudge{P(m)}{C}{P(m+1)}}{\djudge{P(0)}{C^\star}{P(k)}}
  \hva\and
  \inferrule[(\Bd:Local)]{\djudge{P}{C}{Q}}{\djudge{\Sup x.P}{\Local{x}{C}}{\Sup x.Q}}
  \hva\and
  \inferrule[(\Bd:Disj)]{\forall i\in I.\djudge{P_i}{C}{Q_i}}{\djudge{\bigcurlyvee_{i\in I}P_i}{C}{\bigcurlyvee_{i\in I}Q_i}}
  \hva\and
  \inferrule[\Bd:Constancy)]{\djudge{P}{C}{Q}\\\Fv(B)\cap\Mod(C)=\emptyset}{\djudge{P\curlywedge[B]}{C}{Q\curlywedge[B]}}
  \hva\and
  \inferrule[(\Bd:Relax)]{\djudge{P}{C}{Q}\\\Fv(F)\cap\Mod(C)=\emptyset\\F\preceq 0}{\djudge{P+F}{C}{Q+F}}
  \hva\and
  \inferrule[(\Bd:Cons)]{P\preceq P' \\ \djudge{P'}{C}{Q'} \\ Q'\preceq Q}{\djudge{P}{C}{Q}}
  \hva\and
  \inferrule[(\Bd:Subst)]{\djudge{P}{C}{Q}\\y\notin\Fv(P)\cup\Fv(Q)\cup\Fv(C)}{\djudge{P[y/x]}{C[y/x]}{Q[y/x]}}
\end{mathpar}
\caption{Proof Rules for \QBUAd Triples}
\label{fig:fullqbuad}
\end{figure}

\cref{fig:fullqbuad} shows the complete proof rules for \QBUAd triples.

\begin{theorem}[\QBUAd Soundness]\label{thm:qbuad-sound}
  For all $P$, $C$, $Q$, if $\djudge{P}{C}{Q}$ is derivable, then $\dJudge{P}{C}{Q}$ holds.
\end{theorem}

\begin{proof}
  By induction on the derivation.
  \begin{itemize}
    \item \textbf{Case} \textsc{(\Bd:Skip)}:\\
      Let $P$ be arbitrary such that $P\preceq 0$. For any $\sigma$ and $p$ satisfying $p\le P(\sigma)\le 0$, take $\tau=\sigma$ and $q=p$. Then we have $q\le P(\tau)$ and $\bigstep{\Skip}{\sigma}{p}{p}{\tau}{q}$ by \textsc{(BS:Skip)} for $p\le 0$. It follows that $\dJudge{P}{\Skip}{P}$.
    \item \textbf{Case} \textsc{(\Bd:Assign)}:\\
      Let $P$ be arbitrary such that $P\preceq 0$, and let $Q=\Sup x'.P[x'/x]\curlywedge[x=e[x'/x]]$. For any $\sigma$ and $p$ satisfying $p\le P(\sigma)$, take $\tau=\sigma[x\mapsto\eval{e}{\sigma}]$ and $q=p$. Notice that
      \begin{align*}
        Q(\tau)
        &=\sup_v\{(P[x'/x]\curlywedge[x=e[x'/x]])(\tau[x\mapsto v])\}\\
        &=\sup_v\{\min\{(P[x'/x])(\tau[x\mapsto v]),[x=e[x'/x]](\tau[x\mapsto v])\}\}\\
        &=\sup_v\{\min\{P(\tau[x\mapsto v]),[\eval{x}{\tau}=\eval{e}{\tau[x\mapsto v]}]\}\}\\
        &\ge\min\{P(\tau[x\mapsto\eval{x}\sigma]),[\eval{x}{\tau}=\eval{e}{\tau[x\mapsto\eval{\sigma}{x}]}]\}\\
        &=P(\sigma).
    \end{align*}
    Combining with $\bigstep{\Assign{x}{e}}{\sigma}{p}{p}{\tau}{q}$ by \textsc{(BS:Assign)} and $p\le P(\sigma)\le 0$, we have $\dJudge{P}{\Assign{x}{e}}{Q}$.
    \item \textbf{Case} \textsc{(\Bd:Assume)}:\\
      Let $P$ be arbitrary such that $P\preceq 0$. For any $\sigma$ and $p$ satisfying $p\le (P\curlywedge[B])(\sigma)\le 0$, we have $p\le P(\sigma)\le 0$ and $\eval{B}{\sigma}=\true$. Take $\tau=\sigma$ and $q=p$. Then we have $q\le (P\curlywedge[B])(\tau)$ and $\bigstep{\Assume{B}}{\sigma}{p}{p}{\tau}{q}$ by \textsc{(BS:Assume)} for $p\le 0$. It follows that $\dJudge{P\curlywedge[B]}{\Assume{B}}{P\curlywedge[B]}$.
    \item \textbf{Case} \textsc{(\Bd:Tick)}:\\
      Let $P$ be arbitrary such that $(P\curlywedge P-e)\preceq 0$. For any $\sigma$ and $p$ satisfying $p\le P(\sigma)$, take $\tau=\sigma$ and $q=p-\eval{e}{\sigma}$. Then we have $q\le (P-e)(\tau)$ and $\bigstep{\Tick{e}}{\sigma}{p}{\min\{p,p-\eval{e}{\sigma}\}}{\tau}{q}$ by \textsc{(BS:Tick)}. Since $\min\{p,p-\eval{e}{\sigma}\}=(P\curlywedge P-e)(\sigma)\le 0$, we have $\dJudge{P}{\Tick{e}}{P-e}$.
    \item \textbf{Case} \textsc{(\Bd:SeqL)}:\\
      Let $P,Q,R,C_1,C_2$ be arbitrary such that $\djudge{P}{C_1}{R}$ and $\bjudge{R}{C_2}{Q}$. By induction hypothesis, we have $\dJudge{P}{C_1}{R}$ and by theorem \ref{thm:qbua-sound} we have $\bJudge{R}{C_2}{Q}$. For any $\sigma$ and $p$ satisfying $p\le P(\sigma)$, there exists a $\rho$ and an $r$ such that $r\le R(\rho)$ and $\bigstep{C_1}{\sigma}{p}{l_1}{\rho}{r}$ for some $l_1\le 0$. Then there exists a $\tau$ and a $q$ such that $q\le Q(\tau)$ and $\bigstep{C_2}{\rho}{r}{l_2}{\tau}{q}$ for some $l_2$. Hence, $\bigstep{\Seq{C_1}{C_2}}{\sigma}{p}{\min\{l_1,l_2\}}{\tau}{q}$ holds by \textsc{(BS:Seq)}. Since $l_1\le 0$, we have $\min\{l_1,l_2\}\le 0$. It follows that $\dJudge{P}{\Seq{C_1}{C_2}}{Q}$.
    \item \textbf{Case} \textsc{(\Bd:SeqR)}:\\
      This case is similar to the previous one.
    \item \textbf{Case} \textsc{(\Bd:ChoiceL)}:\\
      Let $P,Q,C_1,C_2$ be arbitrary such that $\djudge{P}{C_1}{Q}$. By induction hypothesis, we have $\dJudge{P}{C_1}{Q}$. For any $\sigma$ and $p$ satisfying $p\le P(\sigma)$, there exists a $\tau$ and a $q$ such that $q\le Q(\tau)$ and $\bigstep{C_1}{\sigma}{p}{l}{\tau}{q}$ for some $l\le 0$. Hence, $\bigstep{\Choice{C_1}{C_2}}{\sigma}{p}{l}{\tau}{q}$ holds by \textsc{(BS:ChoiceL)} for $l\le 0$. It follows that $\dJudge{P}{\Choice{C_1}{C_2}}{Q}$.
    \item \textbf{Case} \textsc{(\Bd:ChoiceR)}:\\
      This case is similar to the previous one.
    \item \textbf{Case} \textsc{(\Bd:LoopZero)}:\\
      Let $P$ be arbitrary such that $P\preceq 0$. For any $\sigma$ and $p$ satisfying $p\le P(\sigma)\le 0$, take $\tau=\sigma$ and $q=p$. Then we have $q\le P(\tau)$ and $\bigstep{\Skip}{\Loop{C}}{p}{p}{\tau}{q}$ by \textsc{(BS:LoopZero)} for $p\le 0$. It follows that $\dJudge{P}{\Loop{C}}{P}$.
    \item \textbf{Case} \textsc{(\Bd:Loop)}:\\
      Let $P,C$ be arbitrary such that $\bjudge{P(n)}{C}{P(n+1)}$ for all $n<k$ and $\djudge{P(m)}{C}{P(m+1)}$ for some $m<k$. By theorem \ref{thm:qbua-sound} we have $\bJudge{P(n)}{C}{P(n+1)}$ for all $n<k$ and by induction hypothesis we have $\dJudge{P(m)}{C}{P(m+1)}$. Notice that $\forall m+1\le n<k. \bjudge{P(n)}{C}{P(n+1)}$, so by \textsc{(B:Loop)} we can derive $\bjudge{P(m+1)}{\Loop{C}}{P(k)}$. Thus, we have $\bJudge{P(m+1)}{\Loop{C}}{P(k)}$ by theorem \ref{thm:qbua-sound}.\\
      Therefore, we can proof $\dJudge{P(m)}{\Loop{C}}{P(k)}$ as follows: for any $\sigma$ and $p$ satisfying $p\le P(m)(\sigma)$, by $\dJudge{P(m)}{C}{P(m+1)}$, there exists a $\rho$ and an $r$ such that $r\le P(m+1)(\rho)$ and $\bigstep{C}{\sigma}{p}{l_1}{\rho}{r}$ for some $l_1\le 0$. Then, by $\bJudge{P(m+1)}{\Loop{C}}{P(k)}$, there exists a $\tau$ and a $q$ such that $q\le P(k)(\tau)$ and $\bigstep{\Loop{C}}{\rho}{r}{l_2}{\tau}{q}$ for some $l_2$. Hence, $\bigstep{\Seq{C}{\Loop{C}}}{\sigma}{p}{\min\{l_1,l_2\}}{\tau}{q}$ holds by \textsc{(BS:Seq)}, and consequently $\bigstep{\Loop{C}}{\sigma}{p}{\min\{l_1,l_2\}}{\tau}{q}$ holds by \textsc{(BS:Loop)}. Since $l_1\le 0$, we have $\min\{l_1,l_2\}\le 0$. It follows that $\dJudge{P(m)}{\Loop{C}}{P(k)}$.\\
      The remain thing is to inductively prove that for all $0\le i\le m$, $\dJudge{P(m-i)}{\Loop{C}}{P(k)}$ holds:
      \begin{itemize}
        \item \textbf{Base case} $i=0$:\\
          $\dJudge{P(m)}{\Loop{C}}{P(k)}$ is already proved.
        \item \textbf{Inductive step} $i\to i+1$:\\
          Assume $\dJudge{P(m-i)}{\Loop{C}}{P(k)}$ holds. Let $\sigma$ and $p$ be arbitrary such that $p\le P(m-(i+1))(\sigma)$. By $\bJudge{P(m-(i+1))}{C}{P(m-i)}$, there exists a $\rho$ and an $r$ such that $r\le P(k-i)(\rho)$ and $\bigstep{C}{\sigma}{p}{l_1}{\rho}{r}$ for some $l_1$. Besides, $\dJudge{P(m-i)}{\Loop{C}}{k}$ implies that there exists a $\tau$ and a $q$ such that $q\le P(k)(\tau)$ and $\bigstep{\Loop{C}}{\rho}{r}{l_2}{\tau}{q}$ for some $l_2\le 0$. Hence, $\bigstep{\Seq{C}{\Loop{C}}}{\sigma}{p}{\min\{l_1,l_2\}}{\tau}{q}$ holds by \textsc{(BS:Seq)}, and consequently $\bigstep{\Loop{C}}{\sigma}{p}{\min\{l_1,l_2\}}{\tau}{q}$ holds by \textsc{(BS:Loop)}. Since $l_2\le 0$, we have $\min\{l_1,l_2\}\le 0$. It follows that $\dJudge{P(m-(i+1))}{\Loop{C}}{P(k)}$.
      \end{itemize}
      Taking $i=m$, we have $\dJudge{P(0)}{\Loop{C}}{P(k)}$.
    \item \textbf{Case} \textsc{(\Bd:Local)}:\\
      Let $P,Q,C$ be arbitrary such that $\djudge{P}{C}{Q}$. By induction hypothesis, we have $\dJudge{P}{C}{Q}$. For any $\sigma$ and $p$ satisfying $p\le (\Sup x.P)(\sigma)=\sup_v\{P(\sigma\mapsto v)\}$, there exists a $v$ such that $p\le P(\sigma[x\mapsto v])$. Since $\dJudge{P}{C}{Q}$, there exists a $\tau$ and a $q$ such that $q\le Q(\tau)$ and $\bigstep{C}{\sigma[x\mapsto v]}{p}{l}{\tau}{q}$ for some $l\le 0$. This implies that $q\le (\Sup x.Q)(\tau[x\mapsto\eval{x}{\sigma}]$ and $\bigstep{\Local{x}{C}}{\sigma}{p}{l}{\tau[x\mapsto\eval{x}{\sigma}]}{q}$ by \textsc{(BS:Local)} for $l\le 0$. It follows that $\dJudge{(\Sup x.P)}{\Local{x}{C}}{(\Sup x.Q)}$.
    \item \textbf{Case} \textsc{(\Bd:Constancy)}:\\
      Let $P,Q,B,C$ be arbitrary such that $\djudge{P}{C}{Q}$ and $\Fv(B)\cap\Mod(C)=\emptyset$. By induction hypothesis, we have $\dJudge{P}{C}{Q}$. For any $\sigma$ and $p$ satisfying $p\le (P\curlywedge[B])(\sigma)$, we have $p\le P(\sigma)$ and $\eval{B}{\sigma}=\true$. Then there exists a $\tau$ and a $q$ such that $q\le Q(\tau)$ and $\bigstep{C}{\sigma}{p}{l}{\tau}{q}$ for some $l\le 0$. By proposition \ref{prop:const}, we have $\eval{B}{\tau}=\true$, and then $q\le (Q\curlywedge[B])(\tau)$. It follows that $\dJudge{P\curlywedge[B]}{C}{Q\curlywedge[B]}$.
    \item \textbf{Case} \textsc{(\Bd:Relax)}:\\
      Let $P,Q,F,C$ be arbitrary such that $\djudge{P}{C}{Q}$, $\Fv(F)\cap\Mod(C)=\emptyset$ and $F\preceq 0$. By induction hypothesis, we have $\dJudge{P}{C}{Q}$. For any $\sigma$ and $p$ satisfying $p\le (P+F)(\sigma)=P(\sigma)+F(\sigma)$, we have $p-F(\sigma)\le P(\sigma)$. Then there exists a $\tau$ and a $q$ such that $q\le Q(\tau)$ and $\bigstep{C}{\sigma}{p-F(\sigma)}{l}{\tau}{q}$ for some $l\le 0$. By lemma \ref{lem:relax} and proposition \ref{prop:const}, we have $\bigstep{C}{\sigma}{p}{l+F(\tau)}{\tau}{q+F(\tau)}$. Since $l\le 0$ and $F\preceq 0$, we have $l+F(\tau)\le 0$. It follows that $\dJudge{P+F}{C}{Q}$.
    \item \textbf{Case} \textsc{(\Bd:Cons)}:\\
      Let $P,P',Q,Q',C$ be arbitrary such that $P\preceq P'$, $Q'\preceq Q$ and $\djudge{P'}{C}{Q'}$. By induction hypothesis, we have $\dJudge{P'}{C}{Q'}$. For any $\sigma$ and $p$ satisfying $p\le P(\sigma)$, we have $p\le P'(\sigma)$. Then there exists a $\tau$ and a $q$ such that $q\le Q'(\tau)$ and $\bigstep{C}{\sigma}{p}{l}{\tau}{q}$ for some $l\le 0$. Since $Q'(\tau)\le Q(\tau)$, we have $q\le Q(\tau)$. It follows that $\dJudge{P}{C}{Q}$.
    \item \textbf{Case} \textsc{(\Bd:Subst)}:\\
      Let $P,Q,C$ be arbitrary such that $\djudge{P}{C}{Q}$, and let $x,y$ be arbitrary such that $y\notin\Fv(P)\cup\Fv(Q)\cup\Fv(C)$. By induction hypothesis, we have $\dJudge{P}{C}{Q}$. For any $\sigma$ and $p$ satisfying $p\le P[y/x](\sigma)=P(\sigma[x\mapsto\eval{y}{\sigma}])$, we have $p\le P(\sigma')$ since $y\notin\Fv(P)$, where $\sigma'=\sigma[x\mapsto\eval{y}{\sigma}][y\mapsto\eval{x}{\sigma}]$. Then there exists a $\tau$ and a $q$ such that $q\le Q(\tau)$ and $\bigstep{C}{\sigma'}{p}{l}{\tau}{q}$ for some $l\le 0$. By proposition \ref{prop:subst}, we have $\bigstep{C[y/x]}{\sigma}{p}{l}{\tau'}{q}$ for $l\le 0$, where $\tau'=\tau[x\mapsto\eval{x}{\sigma}][y\mapsto\eval{x}{\tau}]$. Since $y\notin\Fv(Q)$, we have $q\le Q(\tau)=Q[y/x](\tau')$. It follows that $\dJudge{P[y/x]}{C[y/x]}{Q[y/x]}$.
  \end{itemize}
\end{proof}

\begin{lemma}\label{lem:loop-unroll}
  For all $C,\sigma,p,l,\tau,q$, if $\bigstep{\Loop{C}}{\sigma}{p}{l}{\tau}{q}$, then there exists an $n\ge 0$ such that $\bigstep{C^n}{\sigma}{p}{l}{\tau}{q}$.
\end{lemma}

\begin{proof}
  By induction on the derivation of $\bigstep{\Loop{C}}{\sigma}{p}{l}{\tau}{q}$.
  \begin{itemize}
    \item \textbf{Case} \textsc{(BS:LoopZero)}:\\
      In this case, we have $\sigma=\tau$ and $p=l=q$, so we can take $n=0$.
    \item \textbf{Case} \textsc{(BS:Loop)}:\\
      Suppose $\bigstep{\Seq{C}{\Loop{C}}}{\sigma}{p}{l}{\tau}{q}$, then there exist $\rho,r,l_1,l_2$ such that $\bigstep{C}{\sigma}{p}{l_1}{\rho}{r}$, $\bigstep{\Loop{C}}{\rho}{r}{l_2}{\tau}{q}$, and $l=\min\{l_1,l_2\}$. By induction hypothesis, there exists an $n'\ge 0$ such that $\bigstep{C^n}{\rho}{r}{l_2}{\tau}{q}$. Taking $n=n'+1$, we have $\bigstep{C^n}{\sigma}{p}{l}{\tau}{q}$.
  \end{itemize}
\end{proof}

%

\begin{theorem}[QFUA Completeness]\label{thm:qfua-complete}
  For all $P,C,Q$ such that $P$ and $Q$ are finitely supported, if $\fJudge{P}{C}{Q}$ holds, then $\fjudge{P}{C}{Q}$ is derivable.
\end{theorem}

\begin{proof}
  By induction on the structure of $C$.
  \begin{itemize}
    \item \textbf{Case} $\Skip$:\\
      Let $P,Q$ be arbitrary such that $\fJudge{P}{\Skip}{Q}$. From the semantics of the QFUA triple and $\Skip$, we have $P\preceq Q$. Then we can derive $\fjudge{P}{\Skip}{Q}$ using \textsc{(F:Skip)} and \textsc{(F:Cons)}.
    \item \textbf{Case} $\Assign{x}{e}$:\\
      Let $P,Q$ be arbitrary such that $\fJudge{P}{\Assign{x}{e}}{Q}$. From the semantics of the QFUA triple and $\Assign{x}{e}$, we have that for any $\tau$ and $q$ satisfying $q\ge Q(\tau)$, there exists a $\sigma$ such that $q\ge P(\sigma)$ and $\tau=\sigma[x\mapsto\eval{e}{\sigma}]$. Taking $x'=\eval{x}{\sigma}$, we have $q\ge P(\tau[x\mapsto x'])$ and $\eval{x}{\tau}=\eval{e}{\sigma}=\eval{e[x'/x]}{\tau}$. This implies that $q\ge Q'(\tau)$, where $Q'(\tau)=\Inf x'.P[x'/x]\curlyvee[x\ne e[x'/x]]$. Since $q$ is arbitrary, we have $Q'\preceq Q$, and consequently $\fjudge{P}{\Assign{x}{e}}{Q}$ is derivable by \textsc{(F:Assign)} and \textsc{(F:Cons)}.
    \item \textbf{Case} $\Assume{B}$:\\
      Let $P,Q$ be arbitrary such that $\fJudge{P}{\Assume{B}}{Q}$. From the semantics of the QFUA triple and $\Assume{B}$, we have $P\preceq P\curlyvee[\neg B]\preceq Q$. Then we can derive $\fjudge{P}{\Assume{B}}{Q}$ using \textsc{(F:Assume)} and \textsc{(F:Cons)}.
    \item \textbf{Case} $\Tick{e}$:\\
      Let $P,Q$ be arbitrary such that $\fJudge{P}{\Tick{e}}{Q}$. From the semantics of the QFUA triple and $\Tick{e}$, we have $P-e\preceq Q$. Then we can derive $\fjudge{P}{\Tick{e}}{Q}$ using \textsc{(F:Tick)} and \textsc{(F:Cons)}.
    \item \textbf{Case} $\Seq{C_1}{C_2}$:\\
      Let $P,Q$ be arbitrary such that $\fJudge{P}{\Seq{C_1}{C_2}}{Q}$. From the semantics of the QFUA triple and $\Seq{C_1}{C_2}$, we have that for any $\tau$ and $q$ satisfying $q\ge Q(\tau)$, there exists $\sigma,p,\rho,r$ such that $p\ge P(\sigma)$, $\bigstepp{C_1}{\sigma}{p}{\rho}{r}$, and $\bigstepp{C_2}{\rho}{r}{\tau}{q}$. We take $R=\post{C_1}{P}$, then $\fJudge{P}{C_1}{R}$ and $\fJudge{R}{C_2}{Q}$ holds. By induction hypothesis, we have $\fjudge{P}{C_1}{R}$ and $\fjudge{R}{C_2}{Q}$. Thus, we can derive $\fjudge{P}{\Seq{C_1}{C_2}}{Q}$ using \textsc{(F:Seq)}.
    \item \textbf{Case} $\Choice{C_1}{C_2}$:\\
      Let $P,Q$ be arbitrary such that $\fJudge{P}{\Choice{C_1}{C_2}}{Q}$. From the semantics of the QFUA triple and $\Choice{C_1}{C_2}$, we have that for any $\tau$ and $q$ satisfying $q\ge Q(\tau)$, there exists a $\sigma$ and a $p$ such that $p\ge P(\sigma)$ and either $\bigstepp{C_1}{\sigma}{p}{\tau}{q}$ or $\bigstepp{C_2}{\sigma}{p}{\tau}{q}$. So we can divide $Q$ into $Q=Q_1\curlywedge Q_2$, where $Q_i(\tau)=\inf\{q:q\ge Q(\tau)\land\exists\sigma,p.p\ge P(\sigma)\land\bigstepp{C_i}{\sigma}{p}{\tau}{q}\}$ for $i\in\{1,2\}$. By definition, $\fJudge{P}{C_i}{Q_i}$ holds for $i\in\{1,2\}$, which implies that $\fjudge{P}{C_i}{Q_i}$ using induction hypothesis. Then we can derive $\fjudge{P}{\Choice{C_1}{C_2}}{Q}$ using \textsc{(F:Choice)} and \textsc{(F:Disj)}.
    \item \textbf{Case} $\Loop{C}$:\\
      Let $P,Q$ be arbitrary such that $\fJudge{P}{\Loop{C}}{Q}$. From the semantics of the QFUA triple and by lemma \ref{lem:loop-unroll}, we have that for any $\tau$ and $q$ satisfying $q\ge Q(\tau)$, there exists $\sigma,p,n$ such that $p\ge P(\sigma)$ and $\bigstepp{C^n}{\sigma}{p}{\tau}{q}$. So we can divide $Q$ into $Q=\bigcurlywedge_{n\ge 0}Q_n$, where $Q_n(\tau)=\inf\{q:q\ge Q(\tau)\land\exists\sigma,p.p\ge P(\sigma)\land\bigstepp{C^n}{\sigma}{p}{\tau}{q}\}$. In this case, $\fJudge{P}{C^n}{Q_n}$ holds for all $n\ge 0$.\\
      For each $n\ge 0$, let $R_n(i)=\post{C^i}{P}$ for $0\le i<n$ and $R_n(n)=Q_n$.\footnote{By the associativity of sequential composition, $C^i$ can be read as left-associative.} Thus, we have $\fJudge{R_n(i)}C{R_n(i+1)}$ for all $0\le i<n-1$ by definition and $\fJudge{R_n(n-1)}C{R_n(n)}$ by $\fJudge{P}{C^n}{Q_n}$. From induction hypothesis, we have $\fjudge{R_n(i)}C{R_n(i+1)}$ for all $0\le i<n$. Then we can derive $\fjudge{P}{\Loop{C}}{Q_n}$ using \textsc{(F:Loop)}. Finally, we can derive $\fjudge{P}{\Loop{C}}{Q}$ using \textsc{(F:Disj)}.
    \item \textbf{Case} $\Local{x}{C}$:\\
      Let $P,Q$ be arbitrary such that $\fJudge{P}{\Local{x}{C}}{Q}$. Pick $y$ be a fresh variable such that $y\notin\Fv(P)\cup\Fv(Q)\cup\Fv(C)$, and let $Q'=\post{C[y/x]}{P}$. We have $\fJudge{P}{C[y/x]}{Q'}$.
      By induction hypothesis,\footnote{Since $C[y/x]$ is not a subterm of $\Local{x}{C}$, we should use a stronger induction hypothesis: $\fJudge{P}{C[\overrightarrow{y}/\overrightarrow{x}]}{Q}$ implies $\fjudge{P}{C[\overrightarrow{y}/\overrightarrow{x}]}{Q}$, same as incorrectness logic~\cite{POPL:OHearn20}. Formally, we also rely on the semantic equivalence $\Local{y}{C[y/x]} \equiv \Local{x}{C}$, which is captured in our Rocq formalization by an additional rule: $\fjudge{P}{C}{Q}$ and $C\equiv C'$ implies $\fjudge{P}{C'}{Q}$. Similar strengthen applies to the proof of QBUA and \QBUAd completeness.} we have $\fjudge{P}{C[y/x]}{Q'}$, so we can derive $\fjudge{\Inf y.P}{\Local{x}{C}}{\Inf y.Q'}$ using \textsc{(F:Local)}.\\
      Next, we will prove that $\Inf y.Q'\preceq Q$. For any $\tau$ and $q$ satisfying $q\ge Q(\tau)$, since $\fJudge{P}{\Local{x}{C}}{Q}$ holds, there exists a $\sigma$ and a $p$ such that $p\ge P(\sigma)$ and $\bigstepp{\Local{x}{C}}{\sigma}{p}{\tau}{q}$. From the semantics of $\Local{x}{C}$, there exists $\sigma',\tau',v$ such that $\sigma=\sigma'[x\mapsto v]$, $\tau=\tau'[x\mapsto v]$, and $\bigstepp{C}{\sigma'}{p}{\tau'}{q}$. Then we have $\bigstepp{C[y/x]}{\sigma''}{p}{\tau''}{q}$ by proposition \ref{prop:subst}, where $\sigma''=\sigma[y\mapsto\eval{x}{\sigma'}]$ and $\tau''=\tau[y\mapsto\eval{x}{\tau'}]$. Since $y$ is fresh, we have $P(\sigma'')=P(\sigma)$, then we have $q\ge Q'(\tau'')$ from $\bigstepp{C[y/x]}{\sigma''}{p}{\tau''}{q}$. This implies that $q\ge\Inf y.Q'(\tau)$, and consequently $\Inf y.Q'\preceq Q$.\\
      Finally, since $\Inf y.P=P$ by the freshness of $y$, we can derive $\fjudge{P}{\Local{x}{C}}{Q}$ using \textsc{(F:Cons)}.
  \end{itemize}
\end{proof}

\begin{theorem}[QBUA Completeness]\label{thm:qbua-complete}
  For all $P,C,Q$ such that $P$ and $Q$ are finitely supported, if $\bJudge{P}{C}{Q}$ holds, then $\bjudge{P}{C}{Q}$ is derivable.
\end{theorem}

\begin{proof}
  By induction on the structure of $C$.
  \begin{itemize}
    \item \textbf{Case} $\Skip$:\\
      Let $P,Q$ be arbitrary such that $\bJudge{P}{\Skip}{Q}$. From the semantics of the QBUA triple and $\Skip$, we have $P\preceq Q$. Then we can derive $\bjudge{P}{\Skip}{Q}$ using \textsc{(B:Skip)} and \textsc{(B:Cons)}.
    \item \textbf{Case} $\Assign{x}{e}$:\\
      Let $P,Q$ be arbitrary such that $\bJudge{P}{\Assign{x}{e}}{Q}$. Denote $Q'=\Sup x'.P[x'/x]\curlywedge[x=e[x'/x]]$. We will prove that $Q'\preceq Q$. For seek of contradiction, suppose there exists $\tau$ such that $Q'(\tau)>Q(\tau)$. Then there exists $v$ such that $\eval{x}{\tau}=\eval{e}{\tau[x\mapsto v]}$ and $P(\tau[x\mapsto v])>Q(\tau)$. Therefore, take $\sigma=\tau[x\mapsto v]$ and $p=P(\sigma)$. From the semantics of $\bJudge{P}{\Assign{x}{e}}{Q}$, we have $p\le Q(\tau)$, which leads to a contradiction. Thus, we have $Q'\preceq Q$, and consequently $\bjudge{P}{\Assign{x}{e}}{Q}$ is derivable by \textsc{(B:Assign)} and \textsc{(B:Cons)}.
    \item \textbf{Case} $\Assume{B}$:\\
      Let $P,Q$ be arbitrary such that $\bJudge{P}{\Assume{B}}{Q}$. From the semantics of the QBUA triple and $\Assume{B}$, we have $P=P\curlywedge[B]\preceq Q$. Then we can derive $\bjudge{P}{\Assume{B}}{Q}$ using \textsc{(B:Assume)} and \textsc{(B:Cons)}.
    \item \textbf{Case} $\Tick{e}$:\\
      Let $P,Q$ be arbitrary such that $\bJudge{P}{\Tick{e}}{Q}$. From the semantics of the QBUA triple and $\Tick{e}$, we have $P-e\preceq Q$. Then we can derive $\bjudge{P}{\Tick{e}}{Q}$ using \textsc{(B:Tick)} and \textsc{(B:Cons)}.
    \item \textbf{Case} $\Seq{C_1}{C_2}$:\\
      Let $P,Q$ be arbitrary such that $\bJudge{P}{\Seq{C_1}{C_2}}{Q}$. From the semantics of the QBUA triple and $\Seq{C_1}{C_2}$, we have that for any $\sigma$ and $p$ satisfying $p\le P(\sigma)$, there exists $\tau,q,\rho,r$ such that $q\le Q(\tau)$, $\bigstepp{C_1}{\sigma}{p}{\rho}{r}$, and $\bigstepp{C_2}{\rho}{r}{\tau}{q}$. We take $R=\pre{C_2}{Q}$, then $\bJudge{P}{C_1}{R}$ and $\bJudge{R}{C_2}{Q}$ holds. By induction hypothesis, we have $\bjudge{P}{C_1}{R}$ and $\bjudge{R}{C_2}{Q}$. Thus, we can derive $\bjudge{P}{\Seq{C_1}{C_2}}{Q}$ using \textsc{(B:Seq)}.
    \item \textbf{Case} $\Choice{C_1}{C_2}$:\\
      Let $P,Q$ be arbitrary such that $\bJudge{P}{\Choice{C_1}{C_2}}{Q}$. From the semantics of the QBUA triple and $\Choice{C_1}{C_2}$, we have that for any $\sigma$ and $p$ satisfying $p\le P(\sigma)$, there exists a $\tau$ and a $q$ such that $q\le Q(\tau)$ and either $\bigstepp{C_1}{\sigma}{p}{\tau}{q}$ or $\bigstepp{C_2}{\sigma}{p}{\tau}{q}$. So we can divide $P$ into $P=P_1\curlyvee P_2$, where $P_i(\sigma)=\sup\{p:p\le P(\sigma)\land\exists\tau,q.q\le Q(\tau)\land\bigstepp{C_i}{\sigma}{p}{\tau}{q}\}$ for $i\in\{1,2\}$. By definition, $\bJudge{P_i}{C_i}{Q}$ holds for $i\in\{1,2\}$, which implies that $\bjudge{P_i}{C_i}{Q}$ using induction hypothesis. Then we can derive $\bjudge{P}{\Choice{C_1}{C_2}}{Q}$ using \textsc{(B:Choice)} and \textsc{(B:Disj)}.
    \item \textbf{Case} $\Loop{C}$:\\
      Let $P,Q$ be arbitrary such that $\bJudge{P}{\Loop{C}}{Q}$. From the semantics of the QBUA triple and by lemma \ref{lem:loop-unroll}, we have that for any $\sigma$ and $p$ satisfying $p\le P(\sigma)$, there exists $\tau,q,n$ such that $q\le Q(\tau)$ and $\bigstepp{C^n}{\sigma}{p}{\tau}{q}$. So we can divide $P$ into $P=\bigcurlyvee_{n\ge 0}P_n$, where $P_n(\sigma)=\sup\{p:p\le P(\sigma)\land\exists\tau,q.q\le Q(\tau)\land\bigstepp{C^n}{\sigma}{p}{\tau}{q}\}$. In this case, $\bJudge{P_n}{C^n}{Q}$ holds for all $n\ge 0$.\\
      For each $n\ge 0$, let $R_n(i)=\pre{C^{n-i}}{Q}$ for $0<i\le n$ and $R_n(0)=P_n$. Thus, we have $\bJudge{R_n(i-1)}C{R_n(i)}$ for all $1<i\le n$ by definition and $\bJudge{R_n(0)}C{R_n(1)}$ by $\bJudge{P_n}{C^n}{Q}$. From induction hypothesis, we have $\bjudge{R_n(i-1)}C{R_n(i)}$ for all $0<i\le n$. Then we can derive $\bjudge{P_n}{\Loop{C}}{Q}$ using \textsc{(B:Loop)}. Finally, we can derive $\bjudge{P}{\Loop{C}}{Q}$ using \textsc{(B:Disj)}.
    \item \textbf{Case} $\Local{x}{C}$:\\
      Let $P,Q$ be arbitrary such that $\bJudge{P}{\Local{x}{C}}{Q}$. Pick $y$ be a fresh variable such that $y\notin\Fv(P)\cup\Fv(Q)\cup\Fv(C)$, and let $P'=\pre{C[y/x]}{Q}$. We have $\bJudge{P'}{C[y/x]}{Q}$. By induction hypothesis, we have $\bjudge{P'}{C[y/x]}{Q}$, so we can derive $\bjudge{\Sup y.P'}{\Local{x}{C}}{\Sup y.Q}$ using \textsc{(B:Local)}.\\
      Next, we will prove that $P\preceq\Sup y.P'$. For any $\sigma$ and $p$ satisfying $p\le P(\sigma)$, since $\bJudge{P}{\Local{x}{C}}{Q}$ holds, there exists a $\tau$ and a $q$ such that $q\le Q(\tau)$ and $\bigstepp{\Local{x}{C}}{\sigma}{p}{\tau}{q}$. From the semantics of $\Local{x}{C}$, there exists $\sigma',\tau',v$ such that $\sigma=\sigma'[x\mapsto v]$, $\tau=\tau'[x\mapsto v]$, and $\bigstepp{C}{\sigma'}{p}{\tau'}{q}$. Then we have $\bigstepp{C[y/x]}{\sigma''}{p}{\tau''}{q}$ by proposition \ref{prop:subst}, where $\sigma''=\sigma[y\mapsto\eval{x}{\sigma'}]$ and $\tau''=\tau[y\mapsto\eval{x}{\tau'}]$. Since $y$ is fresh, we have $Q(\tau'')=Q(\tau)$, then we have $p\le P'(\sigma'')$ from $\bigstepp{C[y/x]}{\sigma''}{p}{\tau''}{q}$. This implies that $p\le\Sup y.P'(\sigma)$, and consequently $P\preceq\Sup y.P'$.\\
      Finally, since $\Sup y.Q=Q$ by the freshness of $y$, we can derive $\bjudge{P}{\Local{x}{C}}{Q}$ using \textsc{(B:Cons)}.
  \end{itemize}
\end{proof}

\begin{theorem}[\QBUAd Completeness]\label{thm:qbuad-complete}
  For all $P,C,Q$ such that $P$ and $Q$ are finitely supported, if $\dJudge{P}{C}{Q}$ holds, then $\djudge{P}{C}{Q}$ is derivable.
\end{theorem}

\begin{proof}
  By induction on the structure of $C$.
  \begin{itemize}
    \item \textbf{Case} $\Skip$:\\
      Let $P,Q$ be arbitrary such that $\dJudge{P}{\Skip}{Q}$. From the semantics of the \QBUAd triple and $\Skip$, we have $P\preceq Q$ and $P\preceq 0$. Then we can derive $\djudge{P}{\Skip}{Q}$ using \textsc{(\Bd:Skip)} and \textsc{(\Bd:Cons)}.
    \item \textbf{Case} $\Assign{x}{e}$:\\
      Let $P,Q$ be arbitrary such that $\dJudge{P}{\Assign{x}{e}}{Q}$. Denote $Q'=\Sup x'.P[x'/x]\curlywedge[x=e[x'/x]]$. We will prove that $Q'\preceq Q$. For seek of contradiction, suppose there exists $\tau$ such that $Q'(\tau)>Q(\tau)$. Then there exists $v$ such that $\eval{x}{\tau}=\eval{e}{\tau[x\mapsto v]}$ and $P(\tau[x\mapsto v])>Q(\tau)$. Therefore, take $\sigma=\tau[x\mapsto v]$ and $p=P(\sigma)$. From $\dJudge{P}{\Assign{x}{e}}{Q}$, we have $p\le Q(\tau)$, which leads to a contradiction. Thus, we have $Q'\preceq Q$. Besides, we also have $P\preceq 0$ from $\dJudge{P}{\Assign{x}{e}}{Q}$, so $\djudge{P}{\Assign{x}{e}}{Q}$ is derivable by \textsc{(\Bd:Assign)} and \textsc{(\Bd:Cons)}.
    \item \textbf{Case} $\Assume{B}$:\\
      Let $P,Q$ be arbitrary such that $\dJudge{P}{\Assume{B}}{Q}$. From the semantics of the \QBUAd triple and $\Assume{B}$, we have $P=P\curlywedge[B]\preceq Q$ and $P\preceq 0$. Then we can derive $\djudge{P}{\Assume{B}}{Q}$ using \textsc{(\Bd:Assume)} and \textsc{(\Bd:Cons)}.
    \item \textbf{Case} $\Tick{e}$:\\
      Let $P,Q$ be arbitrary such that $\dJudge{P}{\Tick{e}}{Q}$. From the semantics of the \QBUAd triple and $\Tick{e}$, we have $P-e\preceq Q$ and $P\curlywedge P-e\preceq 0$. Then we can derive $\djudge{P}{\Tick{e}}{Q}$ using \textsc{(\Bd:Tick)} and \textsc{(\Bd:Cons)}.
    \item \textbf{Case} $\Seq{C_1}{C_2}$:\\
      Let $P,Q$ be arbitrary such that $\dJudge{P}{\Seq{C_1}{C_2}}{Q}$. From the semantics of the \QBUAd triple and $\Seq{C_1}{C_2}$, we have that for any $\sigma$ and $p$ satisfying $p\le P(\sigma)$, there exists $\tau,q,\rho,r$ such that one of the following holds:
      \begin{itemize}
        \item[(1)] $q\le Q(\tau)$, $\bigstepl{C_1}{\sigma}{p}{\rho}{r}$, and $\bigstepp{C_2}{\rho}{r}{\tau}{q}$;
        \item[(2)] $q\le Q(\tau)$, $\bigstepp{C_1}{\sigma}{p}{\rho}{r}$, and $\bigstepl{C_2}{\rho}{r}{\tau}{q}$.
      \end{itemize}
      So we can divide $P$ into $P=P_1\curlyvee P_2$, where $P_i(\sigma)=\sup\{p:p\le P(\sigma)\land\text{ there exists }\tau,q,\rho,r\text{ such that (}i\text{) holds}\}$ for $i\in\{1,2\}$.
      \begin{itemize}
        \item For the case (1), let $R_1=\pre{C_2}{Q}$, then $\dJudge{P_1}{C_1}{R_1}$ and $\bJudge{R_1}{C_2}{Q}$ holds. By induction hypothesis and theorem \ref{thm:qbua-complete}, we have $\djudge{P_1}{C_1}{R_1}$ and $\bjudge{R_1}{C_2}{Q}$. Thus, we can derive $\djudge{P_1}{\Seq{C_1}{C_2}}{Q}$ using \textsc{(\Bd:SeqL)}.
        \item For the case (2), let $R_2=\prel{C_2}{Q}$, then $\bJudge{P_2}{C_1}{R_2}$ and $\dJudge{R_2}{C_2}{Q}$ holds. By theorem \ref{thm:qbua-complete} and induction hypothesis, we have $\bjudge{P_2}{C_1}{R_2}$ and $\djudge{R_2}{C_2}{Q}$. Thus, we can derive $\djudge{P_2}{\Seq{C_1}{C_2}}{Q}$ using \textsc{(\Bd:SeqR)}.
      \end{itemize}
      Finally, we can derive $\djudge{P}{\Seq{C_1}{C_2}}{Q}$ using \textsc{(\Bd:Disj)}.
    \item \textbf{Case} $\Choice{C_1}{C_2}$:\\
      Let $P,Q$ be arbitrary such that $\dJudge{P}{\Choice{C_1}{C_2}}{Q}$. From the semantics of the \QBUAd triple and $\Choice{C_1}{C_2}$, we have that for any $\sigma$ and $p$ satisfying $p\le P(\sigma)$, there exists a $\tau$ and a $q$ such that $q\le Q(\tau)$ and either $\bigstepl{C_1}{\sigma}{p}{\tau}{q}$ or $\bigstepl{C_2}{\sigma}{p}{\tau}{q}$. So we can divide $P$ into $P=P_1\curlyvee P_2$, where $P_i(\sigma)=\sup\{p:p\le P(\sigma)\land\exists\tau,q.q\le Q(\tau)\land\bigstepl{C_i}{\sigma}{p}{\tau}{q}\}$ for $i\in\{1,2\}$. By definition, $\dJudge{P_i}{C_i}{Q}$ holds for $i\in\{1,2\}$, which implies that $\djudge{P_i}{C_i}{Q}$ using induction hypothesis. Then we can derive $\djudge{P}{\Choice{C_1}{C_2}}{Q}$ using \textsc{(\Bd:Choice)} and \textsc{(\Bd:Disj)}.
    \item \textbf{Case} $\Loop{C}$:\\
      Let $P,Q$ be arbitrary such that $\dJudge{P}{\Loop{C}}{Q}$. From the semantics of the \QBUAd triple and $\Loop{C}$, we have that for any $\sigma$ and $p$ satisfying $p\le P(\sigma)$, there exists $\tau,q,n$ such that $q\le Q(\tau)$ and $\bigstepl{C^n}{\sigma}{p}{\tau}{q}$. So we can divide $P$ into $P=\bigcurlyvee_{n\ge 0}P_n$, where $P_n(\sigma)=\sup\{p:p\le P(\sigma)\land\exists\tau,q.q\le Q(\tau)\land\bigstepl{C^n}{\sigma}{p}{\tau}{q}\}$. In this case, $\dJudge{P_n}{C^n}{Q}$ holds for all $n\ge 0$.\\
      For $n=0$, we have $P_0\preceq 0$ and then $\djudge{P_0}{\Loop{C}}{Q}$ is derivable by \textsc{(\Bd:LoopZero)} and \textsc{(\Bd:Cons)}.
      For each $n>0$, similar to the sequential composition, we can divide $P_n$ into $P_n=\bigcurlyvee_{0\le m<n}P_{n,m}$, where $P_{n,m}(\sigma)=\sup\{p:p\le P_n(\sigma)\land\exists\tau,q,\rho_1,r_1,\rho_2,r_2.q\le Q(\tau)\land\bigstepp{C^m}{\sigma}{p}{\rho_1}{r_1}\land\bigstepl{C}{\rho_1}{r_1}{\rho_2}{r_2}\land\bigstepp{C^{n-m-1}}{\rho_2}{r_2}{\tau}{q}\}$. For each $0\le m<n$, let
      \begin{align*}
        R_{n,m}(i)=\begin{cases}
          P_{n,m} & \text{if }i=0,\\
          \pre{C^{m-i}}{\prel{C}{\pre{C^{n-m-1}}{Q}}} & \text{if }0<i\le m,\\
          \pre{C^{n-i}}{Q} & \text{if }m<i\le n.\\
        \end{cases}
      \end{align*}
      Then we have $\bJudge{R_{n,m}(i)}{C}{R_{n,m}{i+1}}$ for all $i<n$ and $\dJudge{R_{n,m}(m)}{C}{R_{n,m}(m+1)}$. By theorem \ref{thm:qbua-complete} and induction hypothesis, we have $\bjudge{R_{n,m}(i)}{C}{R_{n,m}(i+1)}$ for all $i<n$ and $\djudge{R_{n,m}(m)}{C}{R_{n,m}(m+1)}$. Thus, we can derive $\djudge{P_{n,m}}{C^m}{Q}$ using \textsc{(\Bd:Loop)}, and then derive $\djudge{P_n}{C^n}{Q}$ using \textsc{(\Bd:Disj)}.\\
      Finally, we can derive $\djudge{P}{\Loop{C}}{Q}$ using \textsc{(\Bd:Disj)}.
    \item \textbf{Case} $\Local{x}{C}$:\\
      Let $P,Q$ be arbitrary such that $\dJudge{P}{\Local{x}{C}}{Q}$. Pick $y$ be a fresh variable such that $y\notin\Fv(P)\cup\Fv(Q)\cup\Fv(C)$, and let $P'=\prel{C[y/x]}{Q}$. We have $\dJudge{P'}{C[y/x]}{Q}$. By induction hypothesis, we have $\djudge{P'}{C[y/x]}{Q}$, so we can derive $\djudge{\Sup y.P'}{\Local{x}{C}}{\Sup y.Q}$ using \textsc{(\Bd:Local)}.\\
      Next, we will prove that $P\preceq\Sup y.P'$. For any $\sigma$ and $p$ satisfying $p\le P(\sigma)$, since $\dJudge{P}{\Local{x}{C}}{Q}$ holds, there exists a $\tau$ and a $q$ such that $q\le Q(\tau)$ and $\bigstepl{\Local{x}{C}}{\sigma}{p}{\tau}{q}$. From the semantics of $\Local{x}{C}$, there exists $\sigma',\tau',v$ such that $\sigma=\sigma'[x\mapsto v]$, $\tau=\tau'[x\mapsto v]$, and $\bigstepl{C}{\sigma'}{p}{\tau'}{q}$. Then we have $\bigstepl{C[y/x]}{\sigma''}{p}{\tau''}{q}$ by proposition \ref{prop:subst}, where $\sigma''=\sigma[y\mapsto\eval{x}{\sigma'}]$ and $\tau''=\tau[y\mapsto\eval{x}{\tau'}]$. Since $y$ is fresh, we have $Q(\tau'')=Q(\tau)$, then we have $p\le P'(\sigma'')$ from $\bigstepl{C[y/x]}{\sigma''}{p}{\tau''}{q}$. This implies that $p\le\Sup y.P'(\sigma)$, and consequently $P\preceq\Sup y.P'$.\\
      Finally, since $\Sup y.Q=Q$ by the freshness of $y$, we can derive $\djudge{P}{\Local{x}{C}}{Q}$ using \textsc{(\Bd:Cons)}.
  \end{itemize}
\end{proof}

\end{document}